\documentclass[journal]{IEEEtran} 
\usepackage{graphicx}
\usepackage{graphics}
\graphicspath{ {} }
\usepackage{amsmath}
\usepackage{amsthm}
\usepackage[english]{babel}
\usepackage{amssymb}
\usepackage[usenames, dvipsnames]{color}
\usepackage{cite}

\theoremstyle{plain}
\newtheorem{proposition}{Proposition}
\newtheorem{theorem}{Theorem}
\newtheorem{lemma}{Lemma}
\theoremstyle{remark}
\newtheorem{definition}{Definition}
\newtheorem{remark}{Remark}

\title{Decentralized Heading Control with Rate Constraints using Pulse-Coupled Oscillators}

\author{Timothy Anglea\(^{1}\) (\textit{IEEE Member}) \
	\thanks{$^{1}$Timothy Anglea and Yongqiang Wang are with the department of Electrical \& Computer Engineering, Clemson University, Clemson, SC 29634, USA {\{tbangle,yongqiw\}@clemson.edu}}
	Yongqiang Wang\(^{1}\) (\textit{IEEE Senior Member})}
\date{\today}

\begin{document} 

\maketitle
\begin{abstract}
	Decentralized heading control is crucial for robotic network operations such as surveillance, exploration, and cooperative construction. However, few results consider decentralized heading control when the speed of heading adjustment is restricted. In this paper, we propose a simple hybrid-dynamical model based on pulse-coupled oscillators for decentralized heading control in mobile robots while accounting for the constraint on the rate of heading change. The pulse-coupled oscillator model is effective in coordinating the phase of oscillator networks and hence is promising for robotic heading coordination given that both phase and heading evolve on the same one-dimensional torus. However, existing pulse-coupled oscillator results require the phase adjustment to be instantaneous, which cannot hold for robot heading adjustment due to physical limitations. We propose a generalization to the standard pulse-coupled oscillator model to allow for the phase to adjust at a finite rate, yet still have the oscillator network converge to the desired state, making our approach applicable to robotic heading coordination under rate constraints. We provide rigorous mathematical proof for the achievement of both synchronized and desynchronized heading relationships, and experimentally verify the results using extensive tests on a multi-robot platform.
	
	Index - control constraint, heading control, synchronization, desynchronization, pulse-coupled oscillators, phase response functions
\end{abstract}

\section{Introduction} \label{sec:intro}

Driven by recent accelerated technological advances, mobile robot networks are receiving increased attention, being widely used in warehouse management, surveillance, reconnaissance, search and rescue, and even cooperative construction. Efficient and effective coordination of robotic networks requires advanced cooperative control algorithms, which have been extensively studied in the past decade \cite{bullo2009distributed, Murray_MAS2007, Murray_MVS2007}.

However, due to the disparate discrete-time nature of communication and the continuous-time evolution of motion dynamics, robotic network coordination intrinsically inherits hybrid dynamics, which makes its rigorous treatment very difficult. In fact, to simplify the design, many existing results (e.g., \cite{Sepulchre2007, Paley2007, Sepulchre2008, Chung2009, Moshtagh2009}) choose to design the decentralized controller in the continuous-time domain and then discretize the controller in implementation to conform to the discrete-time nature of communication. Such a design makes the obtained coordination algorithm sensitive to discretization errors, negatively affecting the achievable coordination accuracy.
	
Another difficulty in motion coordination lies in the control input constraint due to a finite speed of heading adjustment, which adds more nonlinearity to the already nonlinear robot network coordination problem.  Note that in the coordination of robot headings, since robot heading evolves in the nonlinear one-dimension torus, existing approaches addressing control-constrained consensus on states that evolve in Euclidean space, such as \cite{Abdessameud2010, YANG2012, zhu2016robust, mo2018distributed, Zhou2018, Deng2018}, are not applicable.

In this paper, we will consider an approach to achieve heading control in a decentralized network with a constraint on the rate of heading change using hybrid dynamics inspired by oscillator networks. Recent research on complex oscillator networks, such as \cite{Gjata2017, Monga2018}, has studied their collective behavior and the ability to control populations of oscillators. However, those results require continuous coupling between oscillators, and do not lend themselves well to the discrete nature of communication in robotic networks. Thus, we will focus on the pulse-coupled oscillator model. Not only is the approach completely decentralized, but it also avoids the discretization problem, as the hybrid dynamics take into account the continuous evolution of motion and the discrete nature of communication explicitly.

The pulse-coupled oscillator (PCO) model was first introduced by Peskin \cite{peskin:75}. He used pulse-coupled oscillators (PCOs) to model the synchronization of pacemaker cells in the heart. Mirollo and Strogatz later improved the model, providing a rigorous mathematical formulation \cite{mirollo:90}. Communication latency, packet loss, signal corruption, and energy consumption are all minimized due to the simplicity of the pulse-based communication between oscillators in the network \cite{Simeone2008}.

In this paper, we apply pulse-coupled oscillator strategies to control the heading, \(\theta\), of a group of mobile robots using the analogy between oscillator phase and robot headings which both evolve in the one-dimensional torus. More specifically, by relating the robot heading to the phase of an oscillator, we can coordinate the robot headings using phase coordination mechanisms based on the pulse-coupled oscillator model. With the many beneficial properties of PCO networks, such as low communication latency and energy consumption \cite{EnergyEfficientSync2012, WangOptimalPRCSync2012}, PCO synchronization and desynchronization strategies seem ideal for the decentralized control of a swarm of mobile robots \cite{GaoMotionControl2018}.

However, existing literature on PCO synchronization and desynchronization strategies, except for our recent work \cite{AngleaClockSync2019}, requires the oscillator phase to change instantaneously \cite{EnergyEfficientSync2012, WangOptimalPRCSync2012, Werner-Allen2005, ScalableSync2005, DorflerSyncSurvey2014, NishimuraSync2015, NunezSync2015, PulseSS2016, Nagpal2007, Scaglione2010, Ferrante2016, AngleaDesync2017, GaoPRCDesync2017}. Due to the natural physical constraint on the rate of rotation, direct application of these strategies is not feasible in actual heading control. In this paper, we propose a generalization of the standard PCO model which does not require instantaneous change in phase to achieve phase coordination. (It is worth noting that, different from our recent work \cite{AngleaClockSync2019} which only addresses clock synchronization, we address both synchronization and desynchronization of PCOs and their application to heading control.) The generalized PCO model can achieve proper phase coordination under an arbitrary constraint on the rate of phase adjustment. To reconcile the constantly evolving phase and the communication-triggered heading update, we also propose a heading coordination mechanism to control robot headings in the presence of heading rate constraints. The results address the hybrid dynamics of continuous-time motion evolution and discrete-time communication directly and hence can be applied in implementation without discretization.

In Section \ref{sec:prelim}, we will discuss preliminary details for a pulse-coupled oscillator network. In Section \ref{sec:HeadingControl}, we will describe the model for mobile robots under heading rate constraints and its relationship to oscillators under phase rate constraints. We will then rigorously show how a mobile robot network under heading rate constraints can achieve heading synchronization and desynchronization in Sections \ref{sec:SyncAnalysis} and \ref{sec:DesyncAnalysis}, respectively. Next, we will present physical experiments on a robotic platform in Section \ref{sec:experiments} that demonstrate the application of pulse-coupled oscillator strategies under the heading rate constraint to collective motion coordination. We will conclude with our results and future work in Section \ref{sec:conclusion}.

\section{Oscillator Preliminaries} \label{sec:prelim} 


Let us consider a network of \(N\) identical pulse-coupled oscillators. Let \(\phi_i \in \mathbb{S}^1 = [0,c_{\phi})\) be the associated phase of oscillator \(i \in \mathcal{V} = \{1,2,\cdots,N\}\), where \(\mathbb{S}^{1}\) is the one-dimensional torus and \(c_{\phi}\) is the phase threshold value for the oscillators. In \cite{EnergyEfficientSync2012, Ferrante2016, AngleaDesync2017, GaoPRCDesync2017}, \(c_{\phi}\) is chosen to be \(2\pi\), but is chosen to be \(1\) in \cite{AngleaClockSync2019, Nagpal2007, Scaglione2010}. In this paper, we will choose \(c_{\phi} = 2\pi\). Each oscillator evolves its phase at a rate \(\omega_i\). For our analysis, each oscillator has an identical fundamental frequency, \(\omega_0\) (i.e., \(\omega_{1} = \omega_{2} = \cdots = \omega_{N} = \omega_{0}\)), and evolves naturally at that rate on the interval \([0,2\pi)\).


During the normal evolution of the network, when an oscillator reaches the threshold value of \(2\pi\), it fires a pulse and resets its phase to zero. All connected oscillators then receive that pulse, being notified of the firing instance of an oscillator in the network. Receiving a pulse will cause an oscillator to change its phase according to the chosen PCO algorithm. Let us denote the amount that oscillator \(i\) adjusts its phase, \(\phi_i\), when a pulse is received as
\begin{equation} \label{eq:PhaseAdjustment}
\psi_i = F(\phi_i) = \lim_{\tau \downarrow 0} \big(\phi_i(t + \tau)\big) - \phi_i(t) = \phi_i^{+} - \phi_i
\end{equation}
where \(\phi_i\) and \(\phi_i^{+}\) represent the phase of oscillator \(i\) before and after receiving a pulse, respectively. The function \(F(\phi_i)\) is called a phase response function (PRF) or phase response curve (PRC), and is used to determine the amount that an oscillator will adjust its phase as a function of the oscillator's phase when a pulse is received \cite{izhikevich2007dynamical}. By an appropriate choice for the PRC and oscillator network topology, the phases of the oscillators in the network can converge to a state of synchronization, where all of the phases achieve the same value \cite{AngleaClockSync2019, EnergyEfficientSync2012} or a state of phase desynchronization, where the phases evenly separate from each other as much as possible \cite{Nagpal2007, Scaglione2010, Ferrante2016, GaoPRCDesync2017}

\section{Pulse-Based Heading Coordination} \label{sec:HeadingControl} 

Let \(\theta_{i}(t) \in \mathbb{S}^{1}\) represent the heading of robot \(i \in \mathcal{V}\) in a network of \(N\) identical robots. To control the heading of the mobile robot network, we propose the following PCO-inspired framework.

In addition to the variable \(\theta_{i}(t)\) to represent the heading of robot \(i\), an internal phase variable \(\phi_{i} \in [0,2\pi)\) is given to each robot such that its initial value is the initial heading of the robot, \(\theta_{i}(0)\). This auxiliary variable represents the phase of an oscillator, and evolves at the rate, or frequency, \(\omega_{i}\). Under no interactions between neighboring robots, the phase variable evolves at the fundamental frequency \(\omega_{i} = \omega_{0}\). When the value of the phase variable \(\phi_{i}\) reaches the threshold of \(2\pi\), it resets to zero, and the robot sends out, or fires, a pulse to all robots connected to it in the network.

The control for the heading of the robot, then, follows the phase variable. The rate at which the robot rotates, \(\dot{\theta}_{i}\), is given by 
\begin{equation} \label{eq:HeadingControl}
\dot{\theta}_{i} = \dot{\phi}_{i} - \omega_{0}
\end{equation}
where \(\dot{\phi_{i}} = \omega_{i}\) is the rate at which the phase variable evolves.

However, there is a physical constraint on how fast the robot can rotate. That is, \(|\dot{\theta}_{i}| \leq \omega_{\max}\) must hold for some maximum rotational rate \(\omega_{\max}\). This natural constraint requires that the heading, \(\theta_{i}(t)\), evolves continuously on the one-dimensional torus. Thus, a constraint is implied on the rate of the phase variable \(\phi_{i}\), such that 
\begin{equation} \label{eq:PhaseRateConstraint} 
|\dot{\theta}_{i}| = |\dot{\phi}_{i} - \omega_{0}| \leq \omega_{\max}
\end{equation}
holds. This implied constraint on the phase variable, \(\phi_{i}\), also requires it to evolve continuously on the one-dimensional torus.

To our knowledge, all existing literature on PCO networks, except for our previous work \cite{AngleaClockSync2019} on clock synchronization, has the phase value change instantaneously, i.e. \(|\dot{\phi}| = \infty\), at pulse firing instances. Since the heading rate constraint implies a rate constraint on the auxiliary phase, we require that the phase not change instantaneously. In this paper, we will generalize the standard PCO model approach such that the phase value is rate constrained at all times, and then apply the results to heading coordination under heading rate constraints.

To ensure the heading rate constraint in \eqref{eq:PhaseRateConstraint}, when robot \(i\) receives a pulse, we let robot \(i\) increase or decrease the rate of evolution of its phase variable, \(\omega_i\), by \(\omega_{\max}\) for a certain amount of time, \(\tau_i\), in order to achieve the required phase adjustment, \(\psi_i\), as in \eqref{eq:PhaseAdjustment}. If the robot receives another pulse before the time \(\tau_i\) is completed, then the robot will use its current phase value \(\phi_i\) to redetermine a new adjustment \(\psi_i\), and thus determine new values for \(\omega_i\) and \(\tau_i\). In fact, the amount of time, \(\tau_i\), the phase variable spends at this new rate is dependent on the phase amount \(\psi_i\) that it needs to adjust, as given in \eqref{eq:FreqMethodTime}.
\begin{equation} \label{eq:FreqMethodTime}
\tau_i = \frac{|\psi_i|}{\omega_{\max}}
\end{equation}

Thus, once an amount of phase adjustment \(\psi_i\) is determined, the phase variable will increase its rate to \(\omega_{i} = \omega_{0} + \omega_{\max}\) if \(\psi_i\) is positive, or decrease its rate to \(\omega_{i} = \omega_{0} - \omega_{\max}\) if \(\psi_i\) is negative, for time \(\tau_i\) determined in \eqref{eq:FreqMethodTime}. Once the time \(\tau_i\) has elapsed, the phase variable returns to its fundamental rate \(\omega_{0}\). If \(\psi_i\) is zero, then the phase variable remains at its fundamental rate, \(\omega_0\), and evolves until the next firing instance.

Therefore, when a robot receives a pulse, its phase variable will adjust its rate according to the rate constraint method described. That is, using the chosen phase response function \(F(\phi_{i})\), the phase variable will begin to evolve at rate \(\omega_{i}\) for time \(\tau_{i}\), and then return to evolving at the fundamental frequency \(\omega_{0}\). Thus, based on the heading control given in \eqref{eq:HeadingControl}, during the normal evolution of the phase variable, the robot does not rotate, and when a pulse is received, the robot will rotate at rate \(\dot{\theta}_{i} = \pm \omega_{\max}\) for time \(\tau_{i}\) and then stop rotating. If a new pulse is received before time \(\tau_{i}\) is complete, then the rate of rotation for the robot will update for the new amount of time \(\tau_{i}\). Fig.1 illustrates the relationship between the robot's phase variable and heading in a network with six robots in an all-to-all topology.

\begin{figure}[t] 
	\includegraphics[width=\columnwidth]{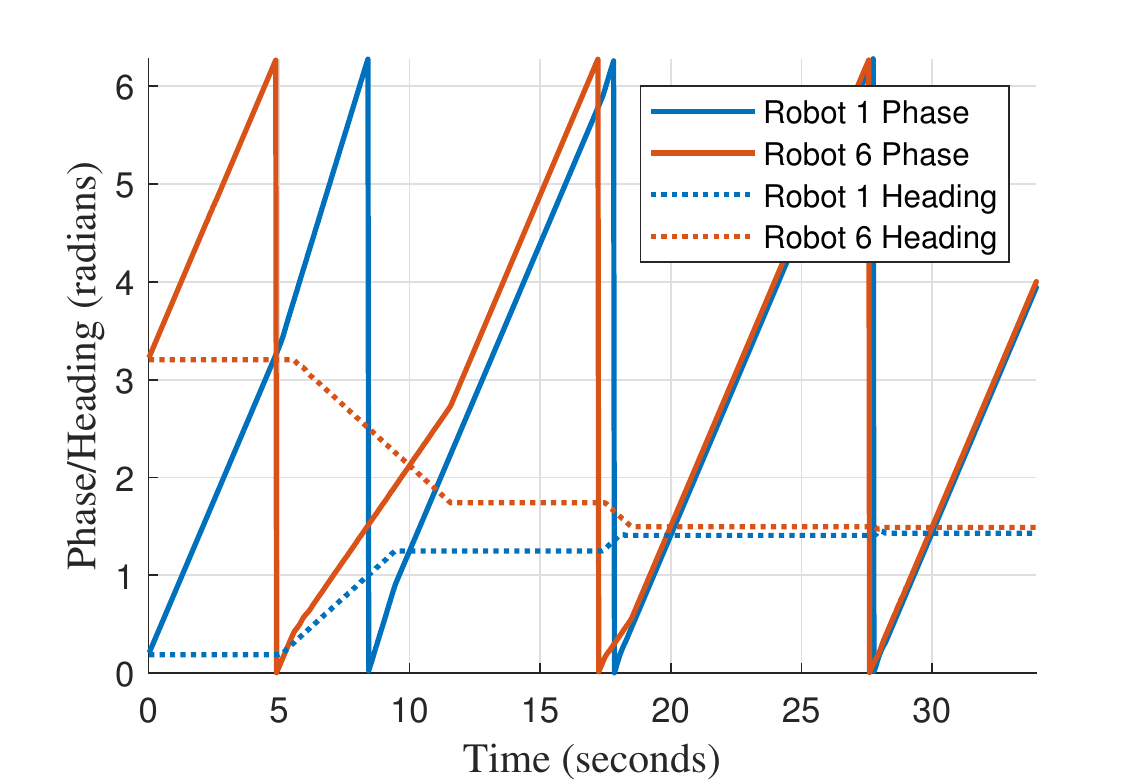}
	\centering
	\caption{Example evolution of two robots' phase and heading values in response to received pulses for a network with \(N=6\) robots in an all-to-all topology. The robot heading changes when the robot phase variable deviates from the fundamental frequency, \(\omega_0\), following the heading control given in \eqref{eq:HeadingControl}.}
	\label{fig:PhaseHeadingComparison}
\end{figure}


With the heading control given in \eqref{eq:HeadingControl}, as the values for the phase variables across all of the robots in the network achieve their desired state by following the chosen PRC, the heading of the mobile robots will also achieve the same desired state. It leaves us to show, then, that the phase variables in the network are able to achieve the desired states of synchronization and desynchronization under the heading rate constraint.

\section{Synchronization under Heading Rate Constraint} \label{sec:SyncAnalysis}

In this section, we consider the dynamics of synchronizing mobile robots under rate constrained heading evolution. All that is required is to take the amount of adjustment for phase variable \(\phi_{i}\), i.e. \(\psi_i\), and determine the necessary amount of time, \(\tau_i\), to change the frequency \(\omega_i\). We then let the phase variable evolve at the new frequency for the required amount of time before it returns to its fundamental frequency, \(\omega _0\).

\begin{figure}[t] 
	\includegraphics[width=\columnwidth]{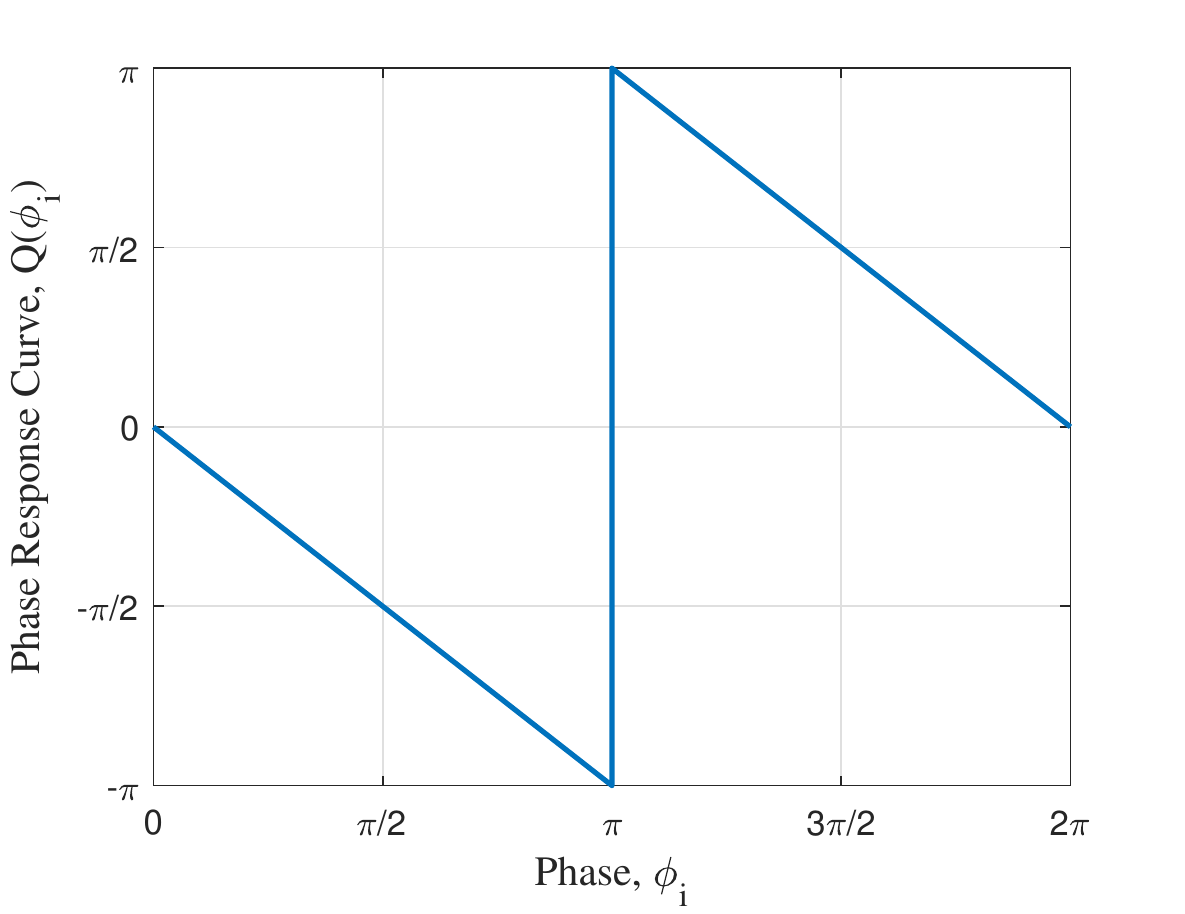}
	\centering
	\caption{The delay-advance phase response curve (PRC) given in (\ref{eq:PRCSyncQ}) for heading synchronization}
	\label{fig:PRCSync1}
\end{figure}
\begin{figure}[t] 
	\includegraphics[width=\columnwidth]{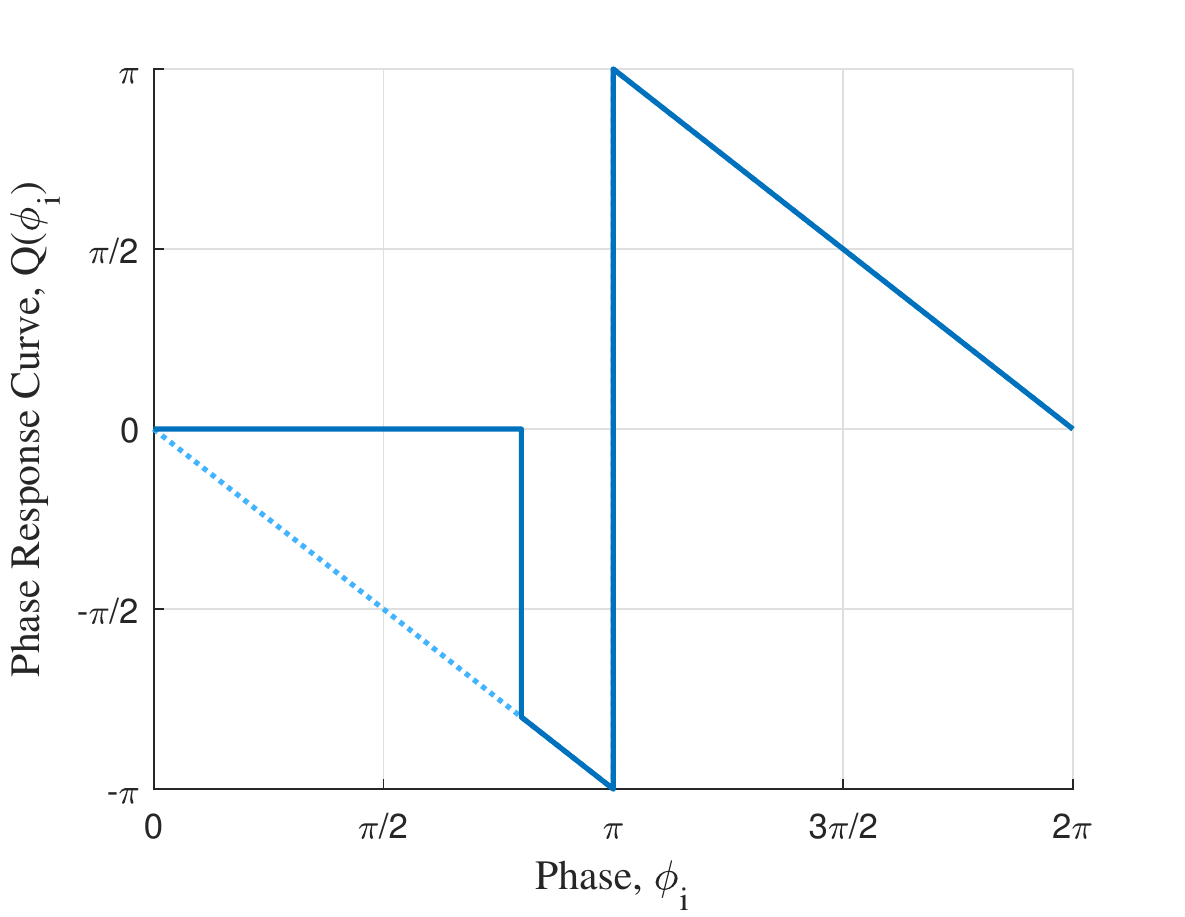}
	\centering
	\caption{The delay-advance phase response curve (PRC) given in (\ref{eq:PRCSyncQ}) for heading synchronization with an example refractory period \(D = 0.8\pi\).}
	\label{fig:PRCSyncRefractory1}
\end{figure}

\subsection{Phase Response Function} 
We will consider the pulse-coupled oscillator PRC synchronization strategy given in \cite{EnergyEfficientSync2012}. This synchronization algorithm uses a delay-advance PRC to describe the phase update at each firing instance. Specifically, when a robot receives a pulse, it updates its phase variable, as in \eqref{eq:PhaseAdjustment}, by the amount given by
\begin{equation} \label{eq:PRFSync}
F(\phi_i) = \alpha Q(\phi_i)
\end{equation}
where
\begin{equation} \label{eq:PRCSyncQ}
Q(\phi_i) = \begin{cases} 
- \phi_{i} & 0 \leq \phi_{i} \leq \pi \\
2\pi - \phi_{i} & \pi < \phi_{i} < 2\pi
\end{cases}
\end{equation}
and \(\alpha \in (0,1]\) is the coupling strength of the network. This PRC is illustrated in Fig. \ref{fig:PRCSync1}. Additionally, a (non-zero) refractory period of duration \(D \in [0, 2\pi)\) can be included in the PRC, such that a robot does not respond to any incoming pulses when its phase variable is within the region of \([0,D)\), and continues to evolve freely, as illustrated in Fig. \ref{fig:PRCSyncRefractory1}. Including this refractory period has been shown to improve energy efficiency within the network \cite{EnergyEfficientSync2012} and to be robust to communication latency \cite{ScalableSync2005}, making this PRC approach a good candidate for achieving synchronization within a mobile robot network.

Without loss of generality, we can order the robots by their initial phase variables such that robot \(1\) has the smallest phase and robot \(N\) the largest (i.e., \(0 \leq \phi_{1} \leq \cdots \leq \phi_{N} < 2\pi\)). To quantify the synchronization within the network, an arc is defined as a connected subset of the interval \([0,2\pi)\). We can thus define the following set of functions, \(v_{i,i+1}\), for all \(i \in \mathcal{V}\):
\color{blue}
\begin{equation} \label{eq:PhaseDiffArcs}
v_{i,i+1}(\phi) = \left\{ \begin{array}{ll}
\phi_{i+1} - \phi_{i} & \mbox{if \(\phi_{i+1} \geq \phi_{i}\)}\\
2\pi - (\phi_{i} - \phi_{i+1}) & \mbox{if \(\phi_{i+1} < \phi_{i}\)} \end{array} \right.
\end{equation}
\color{black}
where robot \(N+1\) maps to robot \(1\). We say that the containing arc of the phase variables is the smallest arc that contains all of the phases in the network. The length of this arc, \(\Lambda\), is given mathematically as

\begin{equation} \label{eq:ContainingArc}
\Lambda = 2\pi - \max_{i \in \mathcal{V}} \{v_{i,i+1}(\phi)\}
\end{equation}

When the phase variables in the network are synchronized, and have the same value, the containing arc, \(\Lambda\), is zero.

\subsection{Synchronization with Heading Rate Constraint}
As shown in \cite{AngleaClockSync2019}, the PRC in \eqref{eq:PRFSync} is capable of achieving synchronization in a connected PCO network under the implied phase rate constraint as described in Sec. \ref{sec:HeadingControl}, where it is possible to model the behavior of the network using a time-varying coupling strength for each phase variable in the network. Specifically, if a new pulse arrives at time \(t_{0} < \tau_{i}\), then the phase variable has only achieved a fraction of the amount of the desired phase change before it responds to the new pulse. This fractional amount of the desired phase change can be seen as a reduction of the coupling strength, \(\alpha\), of robot \(i\) by the ratio \(\frac{t_{0}}{\tau{i}}\), and thus the effective coupling is bounded in the interval \((0,\alpha]\).

We can apply this result to the control of mobile robots.

\begin{theorem} \label{thr:PhaseSyncControl}
	Let \(N\) mobile robots, each with heading \(\theta_{i}\), be in a (strongly) connected network topology, such that the containing arc of the headings is less than some \(\bar{\Lambda} \in (0,\pi]\). Then using the phase response curve given in \eqref{eq:PRFSync} with \(\alpha \in (0,1]\) and refractory period \(D\) not greater than \textcolor{blue}{\(2\pi - \bar{\Lambda}\)}, the heading control given in \eqref{eq:HeadingControl} will synchronize the headings of the robots with the heading rate constraint given in \eqref{eq:PhaseRateConstraint}.
\end{theorem}
\begin{proof}
	The proof for this theorem follows from the results found in \cite{AngleaClockSync2019}. By using the heading control given in \eqref{eq:HeadingControl} such that the initial value for the phase variable \(\phi_{i}\) is set to be \(\theta_{i}(0)\), we have that \(\theta_{i} = (\phi_{i} - \omega_{0}t)\bmod{2\pi}\) for \(i\in \mathcal{V}\). Thus, to show that the headings of the robots \(\theta_{1},\dots,\theta_{N}\) achieve the state of synchronization, we only need to show that the phase variables \(\phi_{1},\dots,\phi_{N}\) achieve the state of synchronization.
	
	Since the initial headings of the robots are in a containing arc \(\bar{\Lambda} \leq \pi\), the initial phase variable values will also be in a containing arc \(\bar{\Lambda} \leq \pi\). From the results in \cite{AngleaClockSync2019}, the phase variables following the phase response curve given in \eqref{eq:PRFSync} under the implied phase rate constraint given in \eqref{eq:PhaseRateConstraint} will synchronize for \(\alpha \in (0,1]\) and refractory period \(D\) not greater than \textcolor{blue}{\(2\pi - \bar{\Lambda}\)}. Therefore, the headings of the robots \(\theta_{1},\dots,\theta_{N}\) will achieve synchronization.
\end{proof}


\section{Desynchronization under Heading Rate Constraint} \label{sec:DesyncAnalysis} 

\begin{figure}[t] 
	\includegraphics[width=\columnwidth]{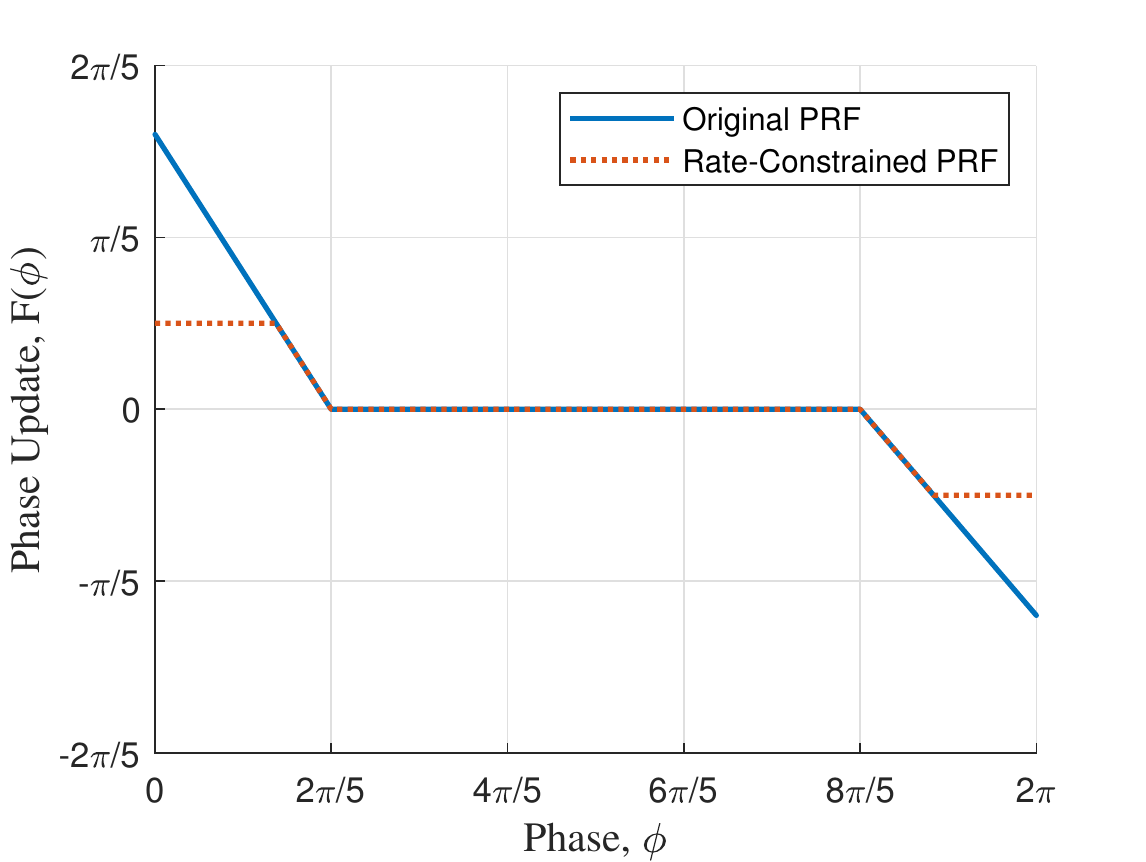}
	\centering
	\caption{Phase response functions (PRFs) to achieve desynchronization under the heading rate constraint. (\(N = 5\), \(l_1 = 0.8\), \(l_2 = 0.6\), \(\omega_0 = 2\pi\), \(\omega_{\max} = 0.5 \omega_0\), \(t_0 = 0.1\) seconds)}
	\label{fig:PRC_continuity}
\end{figure}

\begin{figure}[t] 
	\includegraphics[width=\columnwidth]{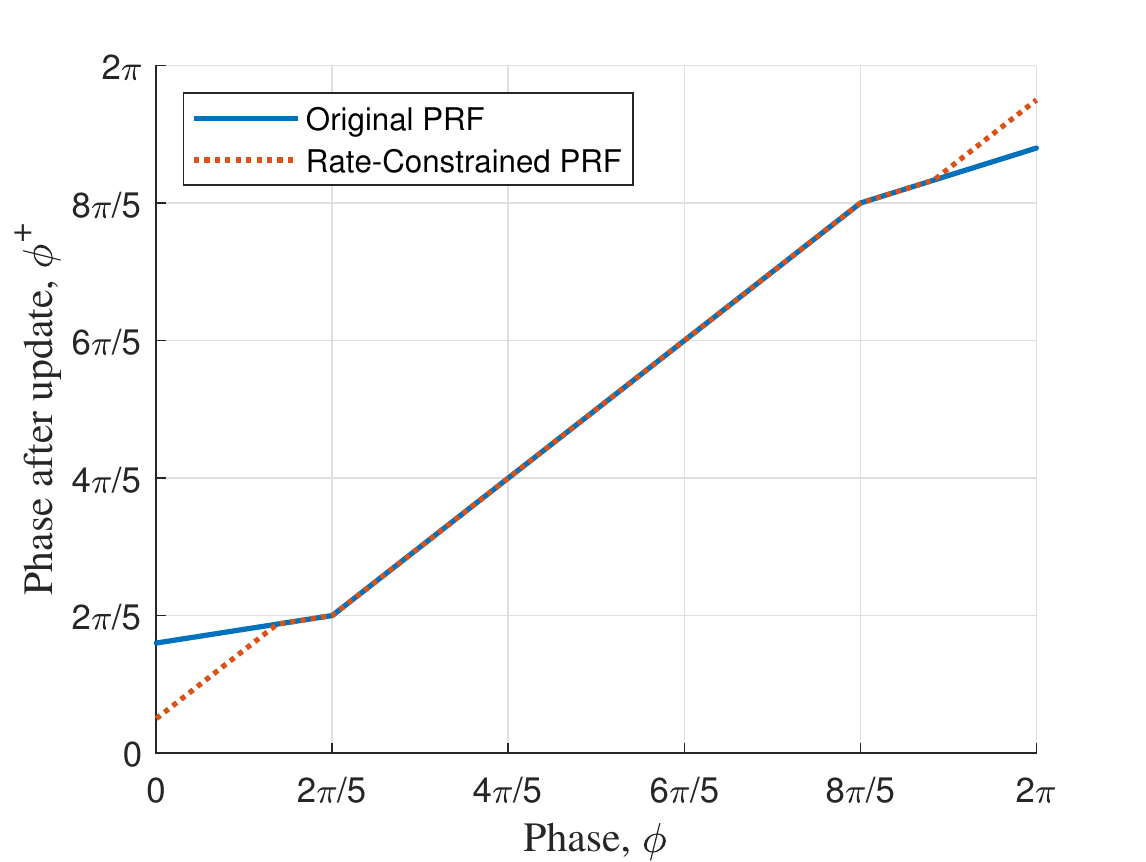}
	\centering
	\caption{Phase update rule based on phase response functions (PRFs) in Fig. \ref{fig:PRC_continuity}. (\(N = 5\), \(l_1 = 0.8\), \(l_2 = 0.6\), \(\omega_0 = 2\pi\), \(\omega_{\max} = 0.5 \omega_0\), \(t_0 = 0.1\) seconds)}
	\label{fig:PhaseUpdate}
\end{figure}

We now rigorously analyze the dynamics of desynchronizing mobile robots with different initial headings under rate constrained heading evolution. As in the previous section on synchronization, all that is required is to take the amount of adjustment for phase variable \(\phi_{i}\), i.e. \(\psi_i\), and determine the necessary amount of time, \(\tau_i\), to change the frequency \(\omega_i\). We then let the phase variable evolve at the new frequency for the required amount of time before it returns to its fundamental frequency, \(\omega _0\).

\subsection{Phase Response Function}
Phase desynchronization is defined to be the network state when all phases variables are spread out evenly such that the difference between any two neighboring phases is exactly \(\frac{2\pi}{N}\). As shown in \cite{GaoPRCDesync2017}, when instantaneous phase changes are allowed, the phase response function (PRF) given by
\begin{equation} \label{eq:PRFGao}
F(\phi_i) = \begin{cases} 
l_{1}(\frac{2\pi}{N} - \phi_{i}) & 0 < \phi_{i} < \frac{2\pi}{N} \\
0 & \frac{2\pi}{N} \leq \phi_{i} \leq 2\pi\frac{N-1}{N} \\
l_{2}(2\pi\frac{N-1}{N} - \phi_{i}) & 2\pi\frac{N-1}{N} < \phi_{i} < 2\pi
\end{cases}
\end{equation}
where \(l_{1} \in [0,1)\) and \(l_{2} \in [0,1)\) are constants denoting the coupling strength such that \(l_{1}\) and \(l_{2}\) are not both zero simultaneously, can guarantee phase desynchronization in an all-to-all connection topology. The function in \eqref{eq:PRFGao} is illustrated with the solid blue line in Fig. \ref{fig:PRC_continuity}. Thus, from \eqref{eq:PhaseAdjustment}, a phase update rule can be written as
\begin{multline} \label{eq:PhaseUpdateGao}
\phi_i^{+} = \phi_i + F(\phi_i) \\ = \begin{cases} 
(1 - l_{1})\phi_{i} + l_{1}\frac{2\pi}{N} & 0 < \phi_{i} < \frac{2\pi}{N} \\
\phi_i & \frac{2\pi}{N} \leq \phi_{i} \leq 2\pi\frac{N-1}{N} \\
(1 - l_{2})\phi_{i} + l_{2}2\pi\frac{N-1}{N} & 2\pi\frac{N-1}{N} < \phi_{i} < 2\pi
\end{cases}
\end{multline}

From \eqref{eq:PhaseUpdateGao}, \(\phi_i^{+} \in (0,2\pi)\) is a strictly monotonically increasing function of \(\phi_i \in (0,2\pi)\), as illustrated with the solid blue line in Fig. \ref{fig:PhaseUpdate}. In this paper, we will use this PRF as a basis for ensuring the phase rate constraint.

To rigorously analyze the network as it converges toward phase desynchronization, we will use the measure given in \cite{GaoPRCDesync2017} to quantify the degree of achievement toward desynchronization. Without loss of generality, we can order the robots by their phase variables such that, at time \(t = 0\), \(0 \leq \phi_N < \phi_{N-1} < \cdots < \phi_2 < \phi_1 < 2\pi\) holds. Since the phase update rule is strictly monotonically increasing, then the firing order of robots is invariant, as shown in Lemma 1 in \cite{GaoPRCDesync2017}. Thus, if robot \(i\) is the next to fire after robot \(i-1\) at time \(t=0\), then it will always be the next to fire after robot \(i-1\). Therefore, the phase variable differences between neighboring robots can always be expressed as:
\begin{equation} \label{eq:PhaseDifference}
\begin{cases} 
\Delta_{i} = (\phi_{i} - \phi_{i+1})\bmod{2\pi}, & i = 1,2,\dots,N-1 \\
\Delta_{N} = (\phi_{N} - \phi_{1})\bmod{2\pi} &  
\end{cases}
\end{equation}

Phase desynchronization implies that the phase differences between neighboring robots are equal to \(\frac{2\pi}{N}\). Thus, the following measure \(P\) is used to quantify phase desynchronization based on phase differences:

\begin{equation} \label{eq:DesyncIndex}
P \triangleq \sum_{i=1}^{N} |\Delta_{i} - \frac{2\pi}{N}|
\end{equation}

It is clear that this measure has a minimum of zero, and equals zero only when phase desynchronization is achieved in the network.

\subsection{Effects of Rate Constraints on Phase Response Function}
Let us analyze the behavior of the robot network under the heading rate constraint given in \eqref{eq:PhaseRateConstraint}. Once robot \(i\) has received a pulse and calculated the necessary change in frequency, \(\omega_i\), to achieve the phase adjustment \(\psi_i\) in time \(\tau_i\), two possibilities can follow: 1) the robot receives no new pulses within time \(\tau_i\) that cause a phase adjustment, and 2) the robot receives a new pulse 
within time \(\tau_i\) that causes a phase adjustment. If we represent the duration of the time interval between the new and previous pulse firing instances as \(t_0\), we can divide these two cases mathematically as 1) \(t_0 \geq \tau_i\), and 2) \(t_0 < \tau_i\).

\begin{enumerate}
	\item In the first case, robot \(i\) finishes adjusting its phase variable by \(\psi_i\), and returns to evolving at the fundamental frequency \(\omega_0\). The same effective change in the phase variable and robot heading has been achieved as if the robot had immediately changed its phase variable and heading by \(\psi_i\) and let the phase variable evolve at rate \(\omega_{0}\) for a time of duration \(t_0\). Thus, no effective change in the phase update rule occurs.
	
	\item In the second case, the robot has not yet achieved its desired amount of change in the phase variable. Rather than having adjusted the whole amount \(\psi_i\), it has adjusted only a portion of that amount, \(\frac{t_0}{\tau_i}\psi_i\), in the time interval between received pulses. The robot then will use its current phase value at the time when the new pulse is received to redetermine new values for \(\psi_i\) and \(\tau_i\). 
	This truncated phase evolution is equivalent to having the robot immediately change its phase variable and heading by \(\frac{t_0}{\tau_i}\psi_i\) and letting the phase variable evolve at rate \(\omega_{0}\) for a time of duration \(t_0\). This fractional amount of the desired change can be seen as a reduction of the coupling strength, \(l_{1}\) or \(l_{2}\), of the PRF for robot \(i\) by the ratio \(\frac{t_0}{\tau_i}\).
\end{enumerate}

We now consider how \eqref{eq:PRFGao} is affected by this behavior caused by the heading rate constraint given in \eqref{eq:PhaseRateConstraint}.
%
The maximum amount of time needed to make the desired phase adjustment is \(\tau_{\max} = \max \{l_{1},l_{2}\} \frac{2\pi}{\omega_{\max} N}\). If, as in the first case above, there are no new pulses within time \(\tau_{\max}\), then there is no effective change to the behavior of the network. Similarly, if the required time, \(\tau_i\), for robot \(i\) to complete the phase adjustment, \(\psi_i\), is smaller than \(t_0\), then there is no effective change to the behavior of that robot. Otherwise, if the required phase adjustment, \(\psi_i\), is larger than \(\omega_{\max} t_0\) (i.e., if \(\tau_i > t_0\)), then each robot with a phase variable \(\phi_i \in (0,\frac{2\pi}{N} - \frac{\omega_{\max} t_0}{l_{1}}) \cup (2\pi\frac{N-1}{N} + \frac{\omega_{\max} t_0}{l_{2}},2\pi)\) when the previous pulse was fired only completes a phase adjustment of \(\omega_{\max} t_0\), a fractional amount of the required adjustment. Thus, for that previous pulse, we can describe the behavior of the network with a PRF given by 
\begin{subequations} \label{eq:PRFFreqMethodAlternate}
\begin{equation} \label{eq:PRF_FreqMethod}
F_{\omega}(\phi_i) = \begin{cases} 
L_{1}(\phi_i)(\frac{2\pi}{N} - \phi_{i}) & 0 < \phi_{i} < \frac{2\pi}{N} \\
0 & \frac{2\pi}{N} \leq \phi_{i} \leq 2\pi\frac{N-1}{N} \\
L_{2}(\phi_i)(2\pi\frac{N-1}{N} - \phi_{i}) & 2\pi\frac{N-1}{N} < \phi_{i} < 2\pi
\end{cases}
\end{equation}
where the forward coupling function, \(L_{1}(\phi_i)\), and the backward coupling function, \(L_{2}(\phi_i)\), are given by
\begin{equation} \label{eq:FCF_FreqMethod} 
L_{1}(\phi_i) = \begin{cases} 
\frac{\omega_{\max} t_0}{\frac{2\pi}{N} - \phi_i} & 0 < \phi_{i} < \frac{2\pi}{N} - \frac{\omega_{\max} t_0}{l_{1}} \\
l_{1} & \frac{2\pi}{N} - \frac{\omega_{\max} t_0}{l_{1}} \leq \phi_{i} < \frac{2\pi}{N}
\end{cases}
\end{equation}
and
\begin{equation} \label{eq:BCF_FreqMethod} 
L_{2}(\phi_i) = \begin{cases} 
l_{2} & 2\pi\frac{N-1}{N} < \phi_{i} \leq 2\pi\frac{N-1}{N} + \frac{\omega_{\max} t_0}{l_{2}} \\
\frac{-\omega_{\max} t_0}{2\pi\frac{N-1}{N} - \phi_i} & 2\pi\frac{N-1}{N} + \frac{\omega_{\max} t_0}{l_{2}} < \phi_{i} < 2\pi
\end{cases}
\end{equation}
\end{subequations}
such that \(\frac{\omega_{\max} t_0}{l_{1}} < \frac{2\pi}{N}\) and \(\frac{\omega_{\max} t_0}{l_{2}} < \frac{2\pi}{N}\). This resulting PRF and the corresponding phase update rule are illustrated with red dotted lines in Fig. \ref{fig:PRC_continuity} and Fig. \ref{fig:PhaseUpdate}, respectively. Note that, in \eqref{eq:PRFFreqMethodAlternate}, \(L_1(\phi_i) \leq l_1\) holds for \(\phi_i \in (0,\frac{2\pi}{N})\), and \(L_2(\phi_i) \leq l_2\) holds for \(\phi_i \in (2\pi\frac{N-1}{N},2\pi)\). Also, though the forward and backward coupling strengths, \(L_1(\phi_i)\) and \(L_2(\phi_i)\), are not constant over their respective domains, the resulting phase update rule, as illustrated in Fig. \ref{fig:PhaseUpdate}, is still strictly monotonically increasing.


\begin{remark} \label{r:RateConstraintApplication}
	The result of applying the heading rate constraint to \eqref{eq:PRFGao} can be interpreted as a PRF with potentially reduced coupling strength functions that are dependent on the value of the phase variable, and that are time-varying based on the timing between received pulses.
\end{remark}

\begin{remark} \label{r:CouplingNotAdjusted}
	The modification to the PRF due to the heading rate constraint described above does not actually adjust the coupling strengths as the network evolves. Only the apparent behavior of the network is being modeled as a PRF with reduced coupling strengths. The actual coupling strengths, \(l_{1}\) and \(l_{2}\), remain unchanged throughout the entire evolution of the network.
\end{remark}

\subsection{Analysis of Phase Response Functions with Time-Varying Coupling Strength Functions}
We now desire to rigorously prove that a network operating under the time-varying phase response function of the form in \eqref{eq:PRFFreqMethodAlternate}, and thus under the heading rate constraint, can achieve desynchronization. To that end, we must prove some preliminary results.

Consider a PRF in the form
\begin{equation} \label{eq:PRFGeneral}
F(\phi_i) = \begin{cases} 
L_{1}(\phi_i)(\frac{2\pi}{N} - \phi_{i}) & 0 < \phi_{i} < \frac{2\pi}{N} \\
0 & \frac{2\pi}{N} \leq \phi_{i} \leq 2\pi\frac{N-1}{N} \\
L_{2}(\phi_i)(2\pi\frac{N-1}{N} - \phi_{i}) & 2\pi\frac{N-1}{N} < \phi_{i} < 2\pi
\end{cases}
\end{equation}
and corresponding phase update rule
\begin{multline} \label{eq:PhaseUpdateGeneral}
\phi_i^{+} = \phi_i + F(\phi_i) \\ = \begin{cases} 
(1 - L_{1}(\phi_{i}))\phi_{i} + L_{1}(\phi_{i})\frac{2\pi}{N} & 0 < \phi_{i} < \frac{2\pi}{N} \\
\phi_i & \frac{2\pi}{N} \leq \phi_{i} \leq 2\pi\frac{N-1}{N} \\
(1 - L_{2}(\phi_{i}))\phi_{i} + L_{2}(\phi_{i})2\pi\frac{N-1}{N} & 2\pi\frac{N-1}{N} < \phi_{i} < 2\pi
\end{cases}
\end{multline}
such that the forward coupling function, \(L_{1}(\phi_i) \in [0,1)\) over the domain \((0,\frac{2\pi}{N})\), and the backward coupling function, \(L_{2}(\phi_i) \in [0,1)\) over the domain \((2\pi\frac{N-1}{N},2\pi)\), both of which may be time-varying, cause the phase update rule in \eqref{eq:PhaseUpdateGeneral} to be strictly increasing over the entire interval \([0,2\pi)\) for all time.

\begin{remark} \label{r:PRFAlterateExpression}
	The PRF given in \eqref{eq:PRFGao} can be written in the form of \eqref{eq:PRFGeneral}, where \(L_{1}(\phi_i) = l_{1}\) and \(L_{2}(\phi_i) = l_{2}\). 
	Additionally, the PRF given in \eqref{eq:PRF_FreqMethod} can be written in the form of \eqref{eq:PRFGeneral}, where \(L_{1}(\phi_{i})\) and \(L_{2}(\phi_{i})\) are given by \eqref{eq:FCF_FreqMethod} and \eqref{eq:BCF_FreqMethod} respectively.
\end{remark}


We must first show that a robotic network under a PRF in the form of \eqref{eq:PRFGeneral} can achieve phase desynchronization. Let us initially consider a robotic network that allows instantaneous change in the phase variable such that the coupling strength functions \(L_{1}(\phi_{i})\) and \(L_{2}(\phi_{i})\) are fixed with respect to time and only vary with respect to the value of the phase variable.

\begin{lemma} \label{lem:PhaseDesyncGeneral}
	For a network of \(N\) robots with no two robots having equal initial heading, and thus equal initial phase variables, the network will achieve phase desynchronization if the phase response function \(F(\phi_i)\) is given by \eqref{eq:PRFGeneral} such that the forward coupling function, \(L_{1}(\phi_i) \in [0,1)\), over the domain \((0,\frac{2\pi}{N})\) and the backward coupling function, \(L_{2}(\phi_i) \in [0,1)\), over the domain \((2\pi\frac{N-1}{N},2\pi)\) cause the phase update rule in \eqref{eq:PhaseUpdateGeneral} to be strictly monotonically increasing.
\end{lemma}
\begin{proof}
	The proof is given in Appendix \ref{ap:PhaseDesyncGeneralProof}. 
\end{proof}

\begin{remark} \label{r:CouplingRestrictions}
	The only condition placed on the forward coupling function, \(L_{1}(\phi_i)\), and the backward coupling function, \(L_{2}(\phi_i)\), is that they cause the resulting phase update rule to be strictly monotonically increasing for \(\phi \in [0,2\pi)\). 
	Any pair of functions, either continuous or discontinuous, that satisfy this condition can be used to form a suitable phase response function.
\end{remark}

Lemma \ref{lem:PhaseDesyncGeneral} proves that a robotic network with a fixed PRF as in \eqref{eq:PRFGeneral} can achieve phase desynchronization under instantaneous changes in the phase variable. We now use this result to show that a PRF that has forward and backward coupling functions that are time-varying, yet still cause the corresponding phase update rule in \eqref{eq:PhaseUpdateGeneral} to be strictly increasing over the interval \([0,2\pi)\) can be used to achieve phase desynchronization.

\begin{proposition} \label{prop:FiringInstanceCoupling}
	The evolution of the phase variable, and thus the heading, of a robot in a robotic network is dependent upon the value of the coupling strengths, \(L_{1}(\phi_i)\) and \(L_{2}(\phi_i)\), only at firing instances.
\end{proposition}
\begin{proof}
	The proof for this proposition is straightforward. A robot only determines the amount that its phase variable needs to change when it receives a pulse. Thus, the values of the coupling strengths are only used at firing instances. Any values the coupling strengths take between firing instances are unused and thus independent of the behavior of the network.
\end{proof}

Using the preliminary results of Lemma \ref{lem:PhaseDesyncGeneral} and Proposition \ref{prop:FiringInstanceCoupling}, we can prove the following theorem.
\begin{theorem} \label{thr:PhaseDesyncTimeVarying}
	For a network of \(N\) robots with no two robots having equal initial heading, and thus equal initial phase variables, the network will achieve phase desynchronization if the phase response function \(F(\phi_i)\) is given by \eqref{eq:PRFGeneral} such that the forward coupling function, \(L_{1}(\phi_i) \in [0,1)\), over the domain \((0,\frac{2\pi}{N})\) and the backward coupling function, \(L_{2}(\phi_i) \in [0,1)\), over the domain \((2\pi\frac{N-1}{N},2\pi)\) can vary over time, yet still cause the phase update rule in \eqref{eq:PhaseUpdateGeneral} to be strictly increasing over the entire interval \([0,2\pi)\).
\end{theorem}
\begin{proof}
	The proof for this theorem is straightforward. From Proposition \ref{prop:FiringInstanceCoupling}, we only need to consider the phase response function at firing instances. At each firing instance, the phase response function fulfills the conditions for Lemma \ref{lem:PhaseDesyncGeneral}. Therefore, the network will achieve phase desynchronization.
\end{proof}


\subsection{Desynchronization with Heading Rate Constraint}
As shown previously, for a robotic network following the control given in \eqref{eq:HeadingControl}, the effect of the heading rate constraint to the standard PRF in \eqref{eq:PRFGao} can be modeled as a PRF in the form of \eqref{eq:PRFGeneral} such that the forward coupling function and the backward coupling function cause the phase update rule in \eqref{eq:PhaseUpdateGeneral} to be strictly increasing over the entire interval \([0,2\pi)\) at each firing instance. Thus, we can present the following main theoretical result.

\begin{theorem} \label{thr:PhaseDesyncControl}
	Let \(N\) mobile robots, each with heading \(\theta_{i}\), be connected in an all-to-all network topology with no two robots having equal initial heading. Then, with the phase response function given by \eqref{eq:PRFGao} for \(l_{1} \in [0,1)\) and \(l_{2} \in [0,1)\), the heading control given in \eqref{eq:HeadingControl} will desynchronize the headings of the robots with the heading rate constraint given in \eqref{eq:PhaseRateConstraint}.
\end{theorem}
\begin{proof}
	The proof for this theorem follows directly from the proof for Theorem \ref{thr:PhaseDesyncTimeVarying}. By using the heading control given in \eqref{eq:HeadingControl} such that the initial value for the phase variable \(\phi_{i}\) is set to be \(\theta_{i}(0)\), we have \(\theta_{i} = (\phi_{i} - \omega_{0}t)\bmod{2\pi}\) for \(i\in \mathcal{V}\). Since no two initial headings of the robots are identical, the initial phase variables for each robot will also be non-identical. Thus, to show that the headings of the robots \(\theta_{1},\dots,\theta_{N}\) achieve the state of phase desynchronization, we only need to show that the phase variables \(\phi_{1},\dots,\phi_{N}\) achieve the state of phase desynchronization. 
	
	Under the heading rate constraint, the resulting phase response function can be described in the form of \eqref{eq:PRFGeneral} such that the forward coupling function and the backward coupling function, which will vary with time due to the timing of successive pulses, cause the phase update rule in \eqref{eq:PhaseUpdateGeneral} to be strictly increasing over the entire interval \([0,2\pi)\). Thus, from Theorem \ref{thr:PhaseDesyncTimeVarying}, the phase variables \(\phi_{1},\dots,\phi_{N}\) will desynchronize following the phase response function given in \eqref{eq:PRFGao} under the heading rate constraint given in \eqref{eq:PhaseRateConstraint}. Therefore, the headings of the robots \(\theta_{1},\dots,\theta_{N}\) will desynchronize.
\end{proof}



\begin{remark} \label{r:IncompletePhaseAdjustmentBackward}
	If a robot fires a pulse, yet has not completed its desired phase change \(\psi_i\) from the previous pulse firing, then the robot must reset its frequency \(\omega_i\) back to its fundamental frequency (until the next pulse firing). Otherwise, the resulting phase update rule may not necessarily be increasing over the interval \([0,2\pi)\), and thus, phase desynchronization is not guaranteed.
\end{remark}

\begin{remark} \label{r:IncompletePhaseAdjustmentForward}
	If a robot receives a pulse when its phase variable is in the interval \([\frac{2\pi}{N},2\pi\frac{N-1}{N}]\), and it has not yet completed its desired phase change \(\psi_i\) from the previous pulse firing (when its phase variable was in the interval \((0,\frac{2\pi}{N})\)), then the robot may continue to evolve at its current frequency \(\omega_i\) determined from the previous pulse firing until it has finished achieving its desired phase change, rather than return to its fundamental frequency. The resulting phase update rule will still be strictly increasing, and the results of Theorem \ref{thr:PhaseDesyncControl} will still apply.
\end{remark}

\begin{remark} \label{r:NonIdenticalRotationRates}
	To achieve heading desynchronization, each robot must change its heading at the same rate, \(\omega_{\max}\), so that the phase update rule in \eqref{eq:PhaseUpdateGeneral} is strictly increasing. If non-identical robots with different maximum rates of rotation are used in the same network, then the network must use some rotation rate that is no larger than the smallest maximum rate.
\end{remark}

\begin{remark} \label{r:AlltoAllCondition}
	It is necessary for the robot network to have an all-to-all topology in order to achieve heading desynchronization. Nearly all existing results on pulse-coupled oscillators require this condition to achieve desynchronization among all oscillators. Relaxing this condition necessitates additional limitations to other aspects of the network, such as initial state \cite{AngleaDesync2017}.
\end{remark}

\begin{remark} \label{r:DifferentInitialHeading}
	It is necessary for no two robots to have identical initial headings in order to achieve heading desynchronization, since they will respond identically to any incoming pulses. To the best of our knowledge, all pulse-based desynchronization algorithms fail in this situation. However, this assumption is not critical, since it is rare for two robots to initialize with the same heading if it is based on analog measurements, such as from a magnetometer. Also, under non-ideal conditions such as pulse propagation delay and pulse drops, it is likely for two robots with identical headings to respond differently to incoming pulses, and thus diverge from each other.
\end{remark}

\section{Experimental Results} \label{sec:experiments}

We will now evaluate the theoretical results from Sections \ref{sec:SyncAnalysis} and \ref{sec:DesyncAnalysis} by performing experiments on a physical robotic platform. With these experiments, we can observe the behavior of the robotic network under practical non-ideal conditions, such as pulse propagation delay and pulse drops.

\subsection{Experimental Setup} 

We perform these experiments on a group of iRobot\textsuperscript{\textregistered} Create 2 programmable robots, or Roombas, as shown in Fig. \ref{fig:RoombaGroup} and Fig. \ref{fig:RoombaSingle}. Roombas provide a relatively simple interface with which to control their actions, along with a wide variety of built-in sensors. The movement of a Roomba is controlled using differential steering, allowing for a simple vehicle model on which to experiment.

Each Roomba is outfitted with a Raspberry Pi 3 Model B microcomputer, which is powered by the Roomba's battery. Programs are written in Python on the Raspberry Pi, allowing the Raspberry Pi to send commands to and receive data back from the Roomba and other sensors. Additionally, a XBee RF communication module is connected to each Raspberry Pi to allow communication between Roombas. To determine the heading of each Roomba, a digital magnetometer can be used to give an initial reading. By using wheel encoder sensor measurements from the Roomba, we track and update the Roomba's heading over time.

\begin{figure}[t] 
	\includegraphics[width=0.85\columnwidth]{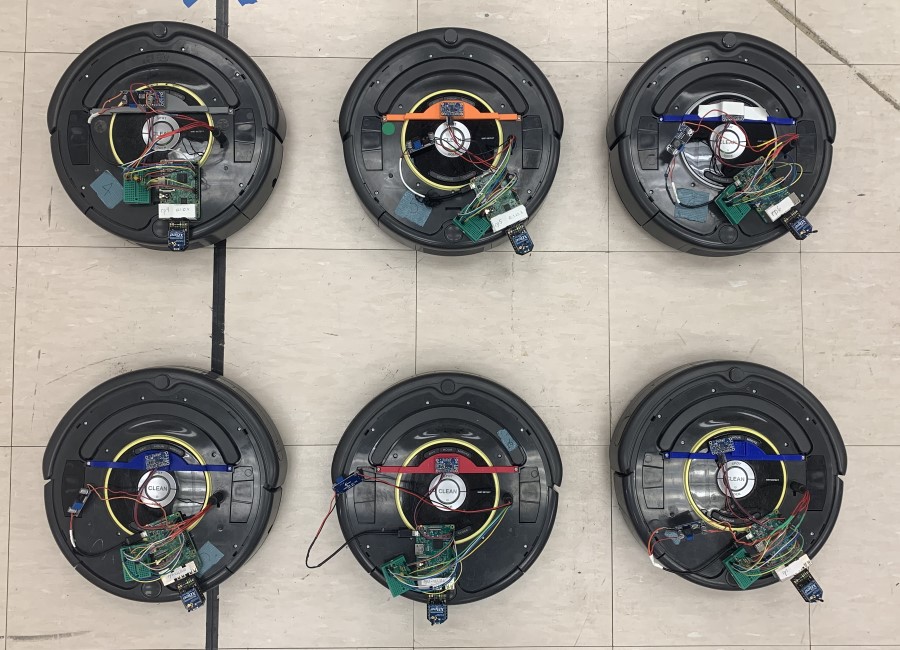}
	\centering
	\caption{The Roomba robotic platform used for experiments. Six Roombas are used to perform the experiments.}
	\label{fig:RoombaGroup}
\end{figure}

\begin{figure}[t] 
	\includegraphics[width=0.65\columnwidth]{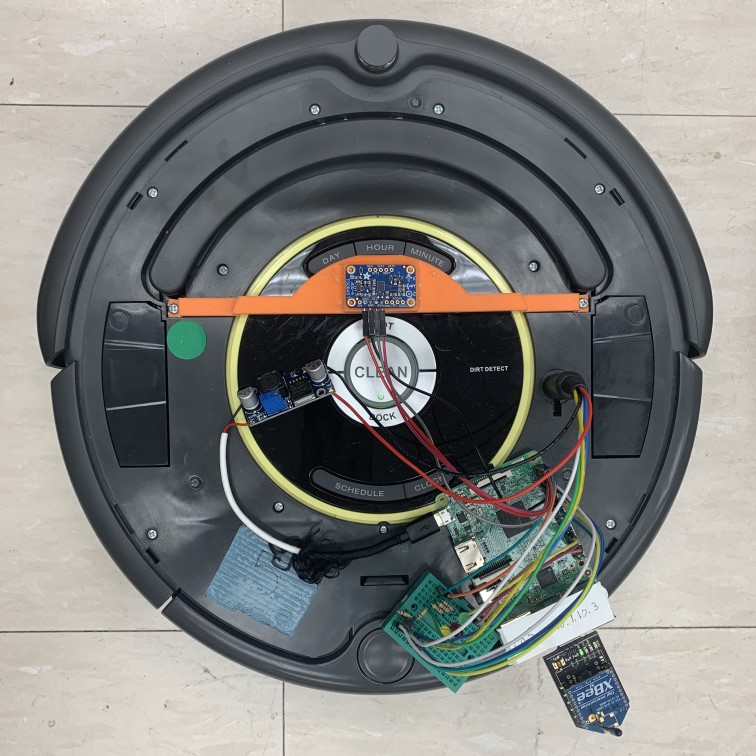}
	\centering
	\caption{A single Roomba used in the experiments. Commands are sent to the Roomba using the Raspberry Pi 3 Model B microcomputer. Pulses are sent between Raspberry Pi 3s using a XBee RF module. Data is retrieved from on-board Roomba sensors to determine heading and orientation.}
	\label{fig:RoombaSingle}
\end{figure}

An oscillator is modeled on each Raspberry Pi, where the phase variable is comprised of the sum of the Roomba's heading and a counter. The counter is based on the internal time of the Raspberry Pi, and increases at the rate \(\omega_{0}\). When a Roomba is not responding to received pulses, the phase variable cycles in the range \([0,2\pi)\) due to the value of the counter while the Roomba is stationary and its heading is unchanging. At the beginning of each experiment, all counters are reset to zero, such that the initial value of the phase variable is equal to the initial heading of the Roomba. To better replicate and compare the results of the experiments, we initialize the heading of each Roomba to a pre-specified direction.

\begin{figure}[t] 
	\includegraphics[width=\columnwidth]{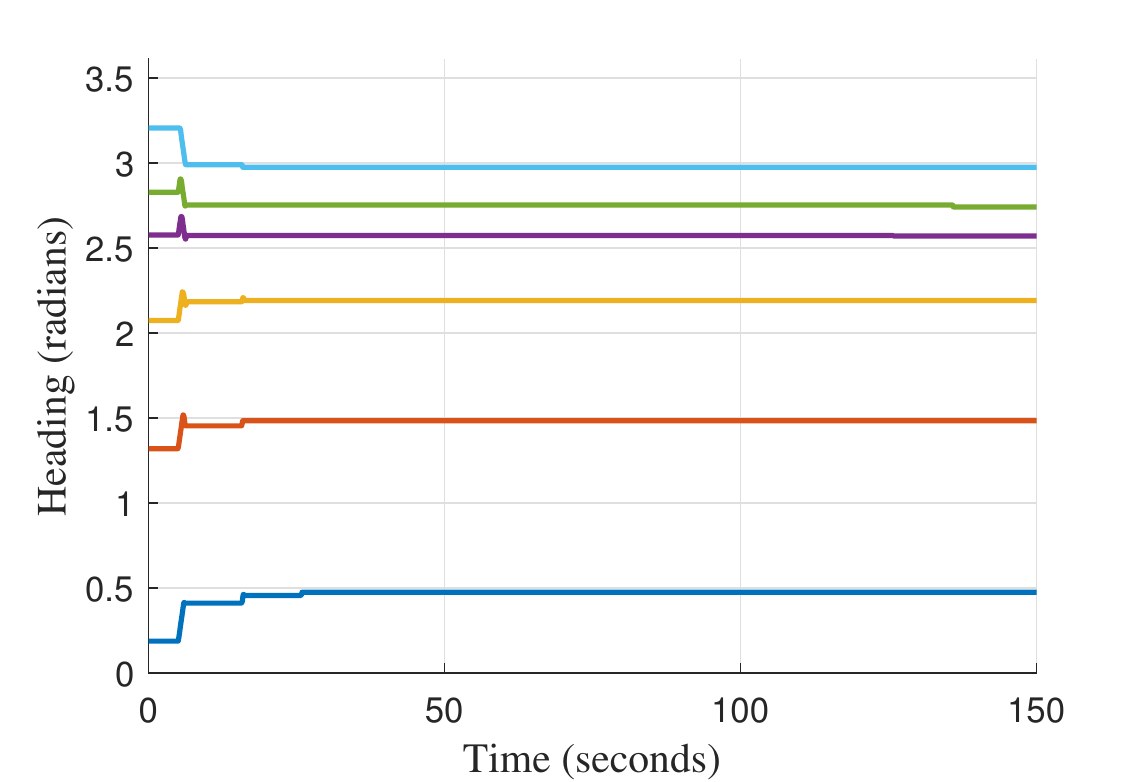}
	\centering
	\caption{Heading evolution of the robots under the heading synchronization algorithm for \(N=6\) robots in an all-to-all topology when the heading rate constraint is ignored, with \(\alpha = 0.5\), \(\omega_0 = \frac{\pi}{5}\), and \(\omega_{\max} = 0.3 \omega_0\). When the heading adjustment is assumed to be instantaneous, the robots are unable to converge their headings.}
	\label{fig:SyncCFM_Heading_Bad}
\end{figure}

Since the RF communication module has a very large communication range (approximately \(100\) meters) compared to the spacing between the Roombas (approximately \(10\) meters), every Roomba can respond to every other Roomba. To test different network topologies, we can define the connections between Roombas by letting each Roomba respond to pulses from a specific set of Roombas. For a single Roomba, when the phase variable reaches the threshold \(2\pi\), a pulse is sent to all connected Roombas through the RF communication module, and the counter value for that Roomba is decreased by \(2\pi\) such that the phase variable is reset to zero. When a Roomba receives a pulse through the RF communication module, the current phase (i.e., the sum of the current heading and counter) is used to determine the amount of phase change required by the PCO algorithm, as given in \eqref{eq:PhaseAdjustment}. This phase change is then applied to the heading of the Roomba. The Roomba spins at a constant rate \(\omega_{\max}\) toward the heading at which it should be until the desired heading change is achieved or another pulse is received. Thus, the control of the heading, as given in \eqref{eq:HeadingControl} with the necessary heading constraint, is satisfied for each Roomba.

The constraint on the rotation rate cannot be ignored when applying the PCO algorithm to physical systems. Because of the heading rate constraint, the robot may not have completed its desired heading change before a new pulse is received. If it is assumed that the robot has achieved the desired heading change when a new pulse is received, even if it has not, then the robot network will not achieve the desired state. We demonstrate this behavior in Fig. \ref{fig:SyncCFM_Heading_Bad}, showing the evolution of the Roomba headings assuming instantaneous heading adjustment according to the PRC in Fig. \ref{fig:PRCSync1}. The headings, rather than synchronizing to the same value, stop adjusting and do not converge. This result confirms the necessity to address heading coordination under rate constraints explicitly.

\subsection{Heading Synchronization} 

We now show the experimental results for a group of six Roombas following the synchronization algorithm detailed in Section \ref{sec:SyncAnalysis}.

\begin{figure}[t] 
	\includegraphics[width=\columnwidth]{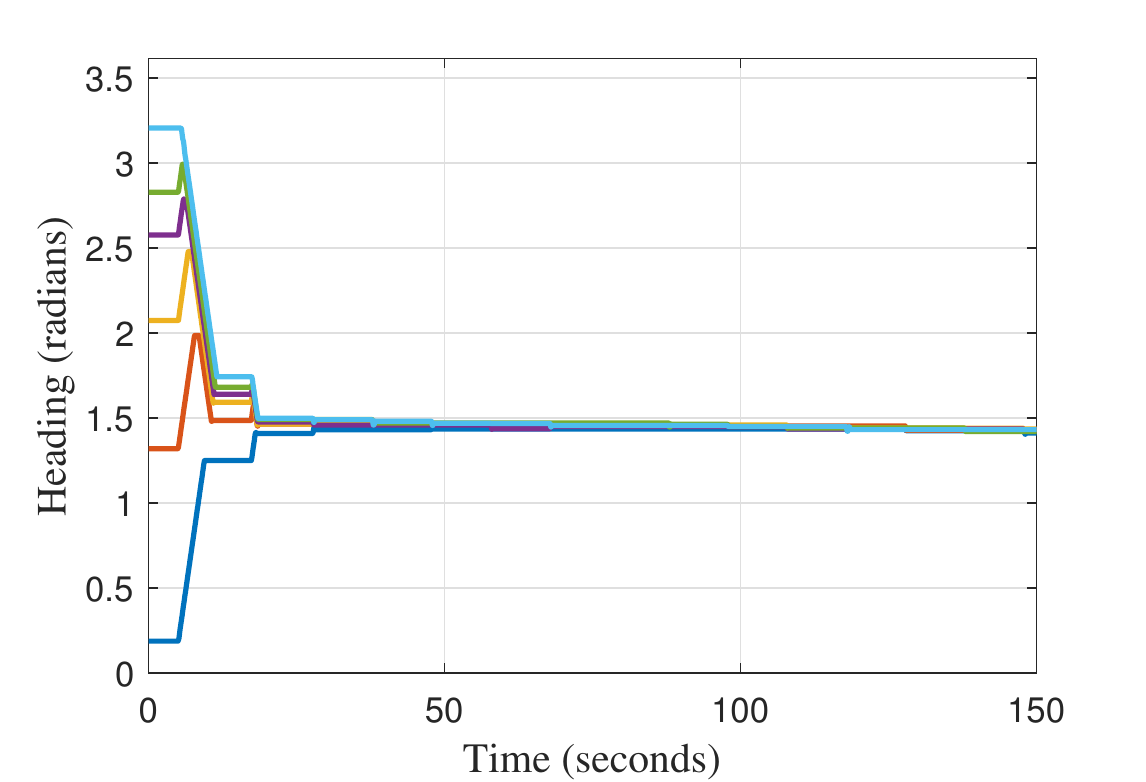}
	\centering
	\caption{Heading evolution of the robots under the heading synchronization algorithm for \(N=6\) robots in an all-to-all topology, with \(\alpha = 0.5\), \(D = 0\), \(\omega_0 = \frac{\pi}{5}\), and \(\omega_{\max} = 0.3 \omega_0\). }
	\label{fig:SyncCFM_Heading}
\end{figure}

\begin{figure}[t] 
	\includegraphics[width=\columnwidth]{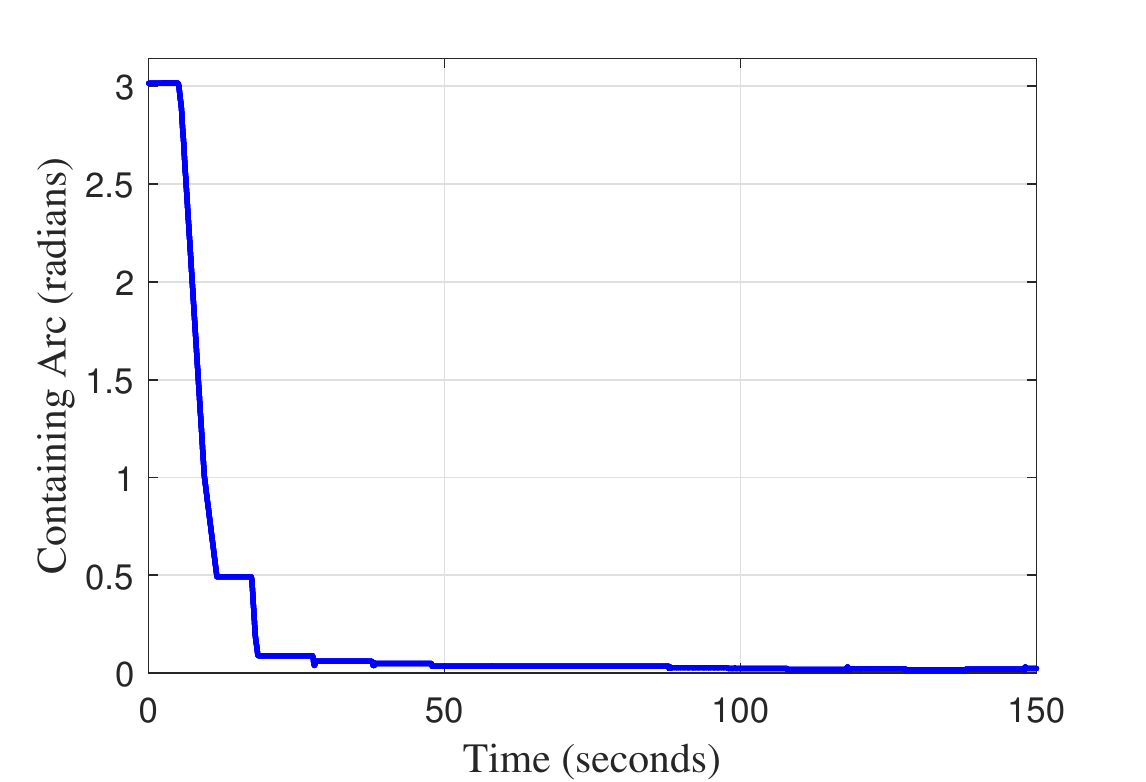}
	\centering
	\caption{Containing arc, \(\Lambda\), as a function of time for the robot headings in Fig. \ref{fig:SyncCFM_Heading}.}
	\label{fig:SyncCFM_ContainingArc}
\end{figure}

\begin{figure}[t] 
	\includegraphics[width=\columnwidth]{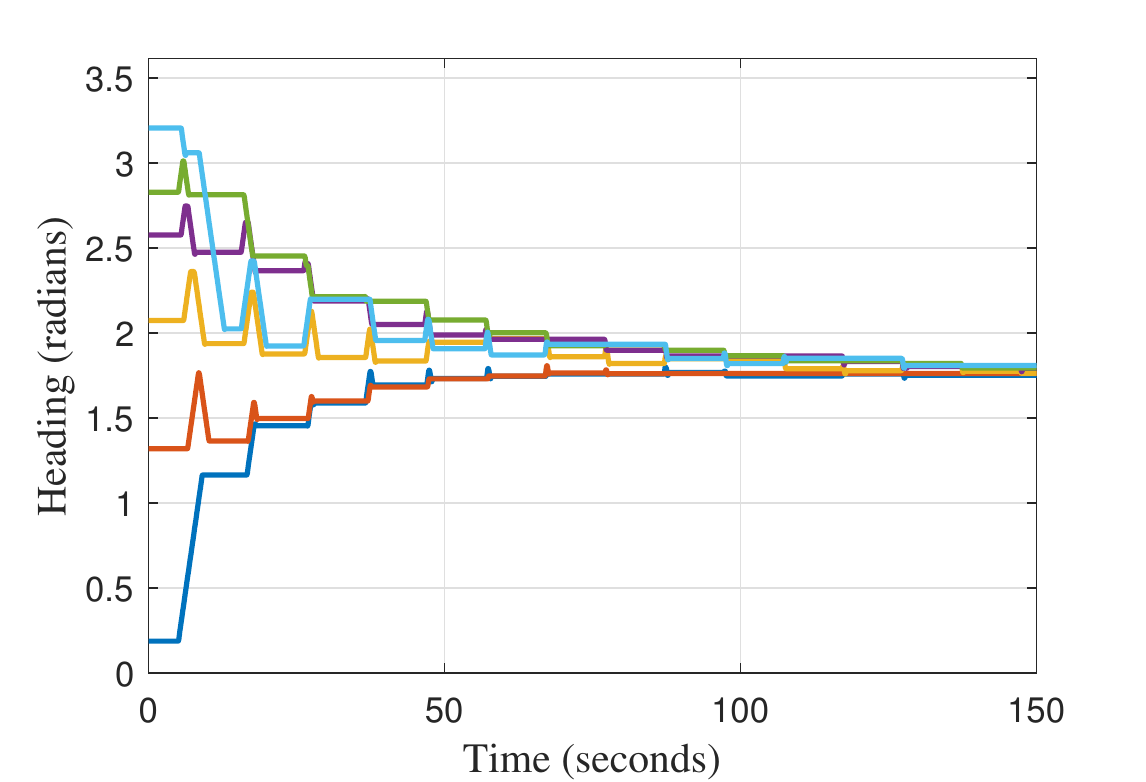}
	\centering
	\caption{Heading evolution of the robots under the heading synchronization algorithm for \(N=6\) robots in a bidirectional ring topology, with \(\alpha = 0.5\), \(D = 0\), \(\omega_0 = \frac{\pi}{5}\), and \(\omega_{\max} = 0.3 \omega_0\).}
	\label{fig:SyncCFM_HeadingRing}
\end{figure}

\begin{figure}[t] 
	\includegraphics[width=\columnwidth]{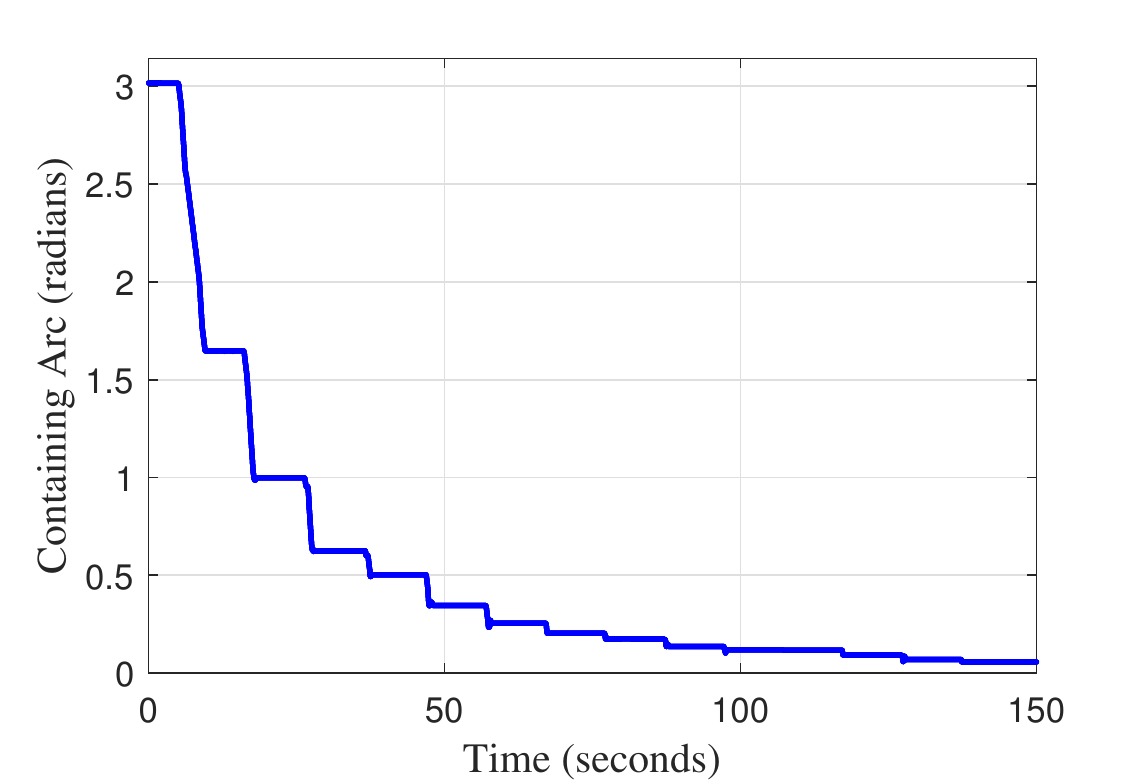}
	\centering
	\caption{Containing arc, \(\Lambda\), as a function of time for the robot headings in Fig. \ref{fig:SyncCFM_HeadingRing}.}
	\label{fig:SyncCFM_ContainingArcRing}
\end{figure}

\begin{figure}[t] 
	\includegraphics[width=\columnwidth]{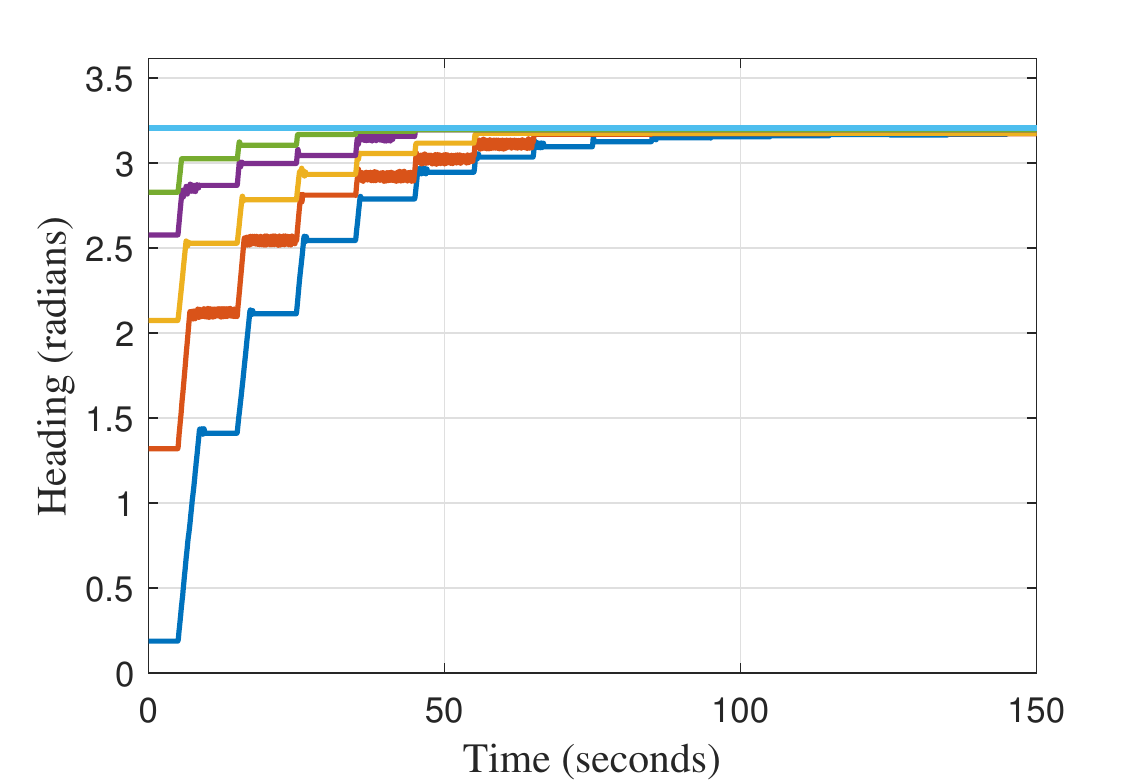}
	\centering
	\caption{Heading evolution of the robots under the heading synchronization algorithm for \(N=6\) robots in an all-to-all topology, with \(\alpha = 0.5\), \(D = \pi\), \(\omega_0 = \frac{\pi}{5}\), and \(\omega_{\max} = 0.3 \omega_0\). }
	\label{fig:SyncCFM_Heading_Refractory}
\end{figure}

\begin{figure}[t] 
	\includegraphics[width=\columnwidth]{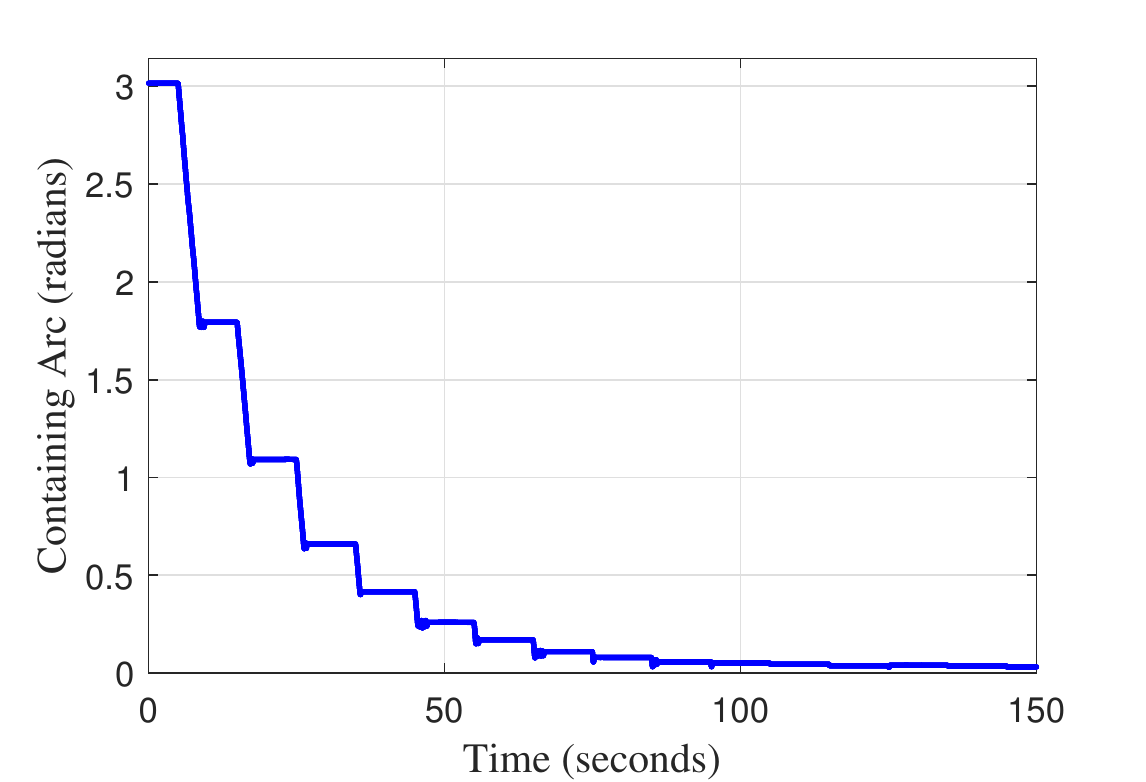}
	\centering
	\caption{Containing arc, \(\Lambda\), as a function of time for the robot headings in Fig. \ref{fig:SyncCFM_Heading_Refractory}.}
	\label{fig:SyncCFM_ContainingArc_Refractory}
\end{figure}

The Roombas, under an all-to-all topology, begin with initial headings within a containing arc \(\Lambda < \pi\). The headings of the Roombas over time are shown in Fig. \ref{fig:SyncCFM_Heading}, in which we see that the headings of the Roombas converge to the same value. As can be seen from the resulting containing arc over time, as shown in Fig. \ref{fig:SyncCFM_ContainingArc}, the network reaches the synchronized state and is stable.

Synchronization of robot headings can also be achieved under more general topologies. We illustrate this point by using the bidirectional ring topology. The headings of the Roombas over time are shown in Fig. \ref{fig:SyncCFM_HeadingRing}, in which we see that the headings of the Roombas converge to the same value. It is apparent from the resulting containing arc shown in Fig. \ref{fig:SyncCFM_ContainingArcRing} that the network synchronizes more slowly with the ring topology compared to the all-to-all topology.

Due to the propagation delay of pulses between Roombas, the headings of the Roomba slowly shift backward over time. This behavior can be reduced by introducing a non-zero refractory period, \(D\), such that robots do not respond to pulses received when their individual phase variable is in the region \([0,D)\) \cite{EnergyEfficientSync2012}. Using the same initial headings for the Roombas, Fig. \ref{fig:SyncCFM_Heading_Refractory} shows the evolution of the network with refractory period \(D = \pi\). As can be seen from the resulting containing arc evolution in Fig. \ref{fig:SyncCFM_ContainingArc_Refractory}, the network reaches the synchronized state and is stable.

\subsection{Heading Desynchronization} 

We next show the experimental results for a group of six Roombas in an all-to-all topology following the desynchronization algorithm detailed in Section \ref{sec:DesyncAnalysis}.

We initialize the Roombas to have headings close together. The evolution of the headings is shown in Fig. \ref{fig:DesyncCFM_Heading}, in which we see that the headings of the Roombas diverge, or spread out equally. By observing the desynchronization measure, given in \eqref{eq:DesyncIndex}, as shown in Fig. \ref{fig:DesyncCFM_ContainingArc}, the network reaches the desynchronized state, and is stable.

Similar to the synchronization case, the Roomba headings slowly drift backward over time. This behavior is again due to the propagation delay of the pulses, and can be mitigated by reducing the backward coupling strength, \(l_{2}\). This mitigation is confirmed by Fig. \ref{fig:DesyncCFM_Heading_L2=0}, which shows the evolution of the robot headings with backward coupling strength \(l_{2} = 0\). The corresponding evolution of the desynchronization measure is shown in Fig. \ref{fig:DesyncCFM_ContainingArc_L2=0}.

\begin{figure}[t] 
	\includegraphics[width=\columnwidth]{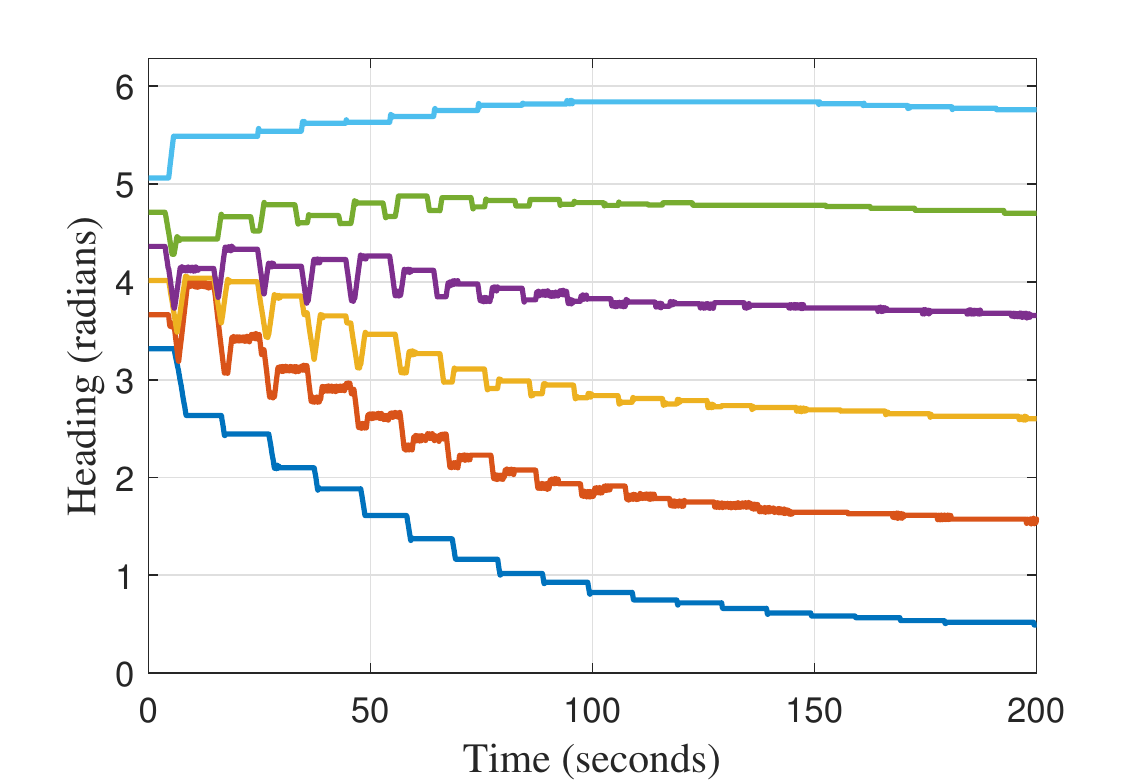}
	\centering
	\caption{Heading evolution of the robots under the heading desynchronization algorithm for \(N=6\) robots in an all-to-all topology, with \(l_{1} = 0.8\), \(l_{2} = 0.6\), \(\omega_0 = \frac{\pi}{5}\), and \(\omega_{\max} = 0.3 \omega_0\).}
	\label{fig:DesyncCFM_Heading}
\end{figure}

\begin{figure}[t] 
	\includegraphics[width=\columnwidth]{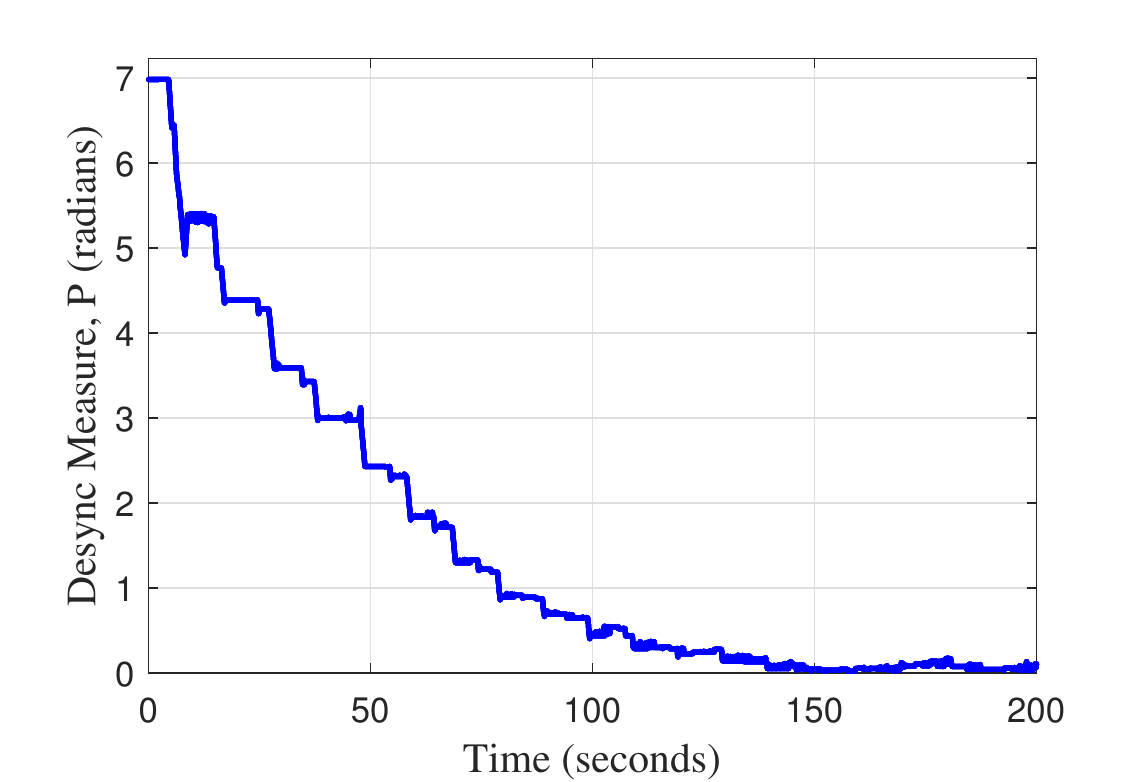}
	\centering
	\caption{Desynchronization measure, \(P\), as a function of time for the robot headings in Fig. \ref{fig:DesyncCFM_Heading}.}
	\label{fig:DesyncCFM_ContainingArc}
\end{figure}

\begin{figure}[t] 
	\includegraphics[width=\columnwidth]{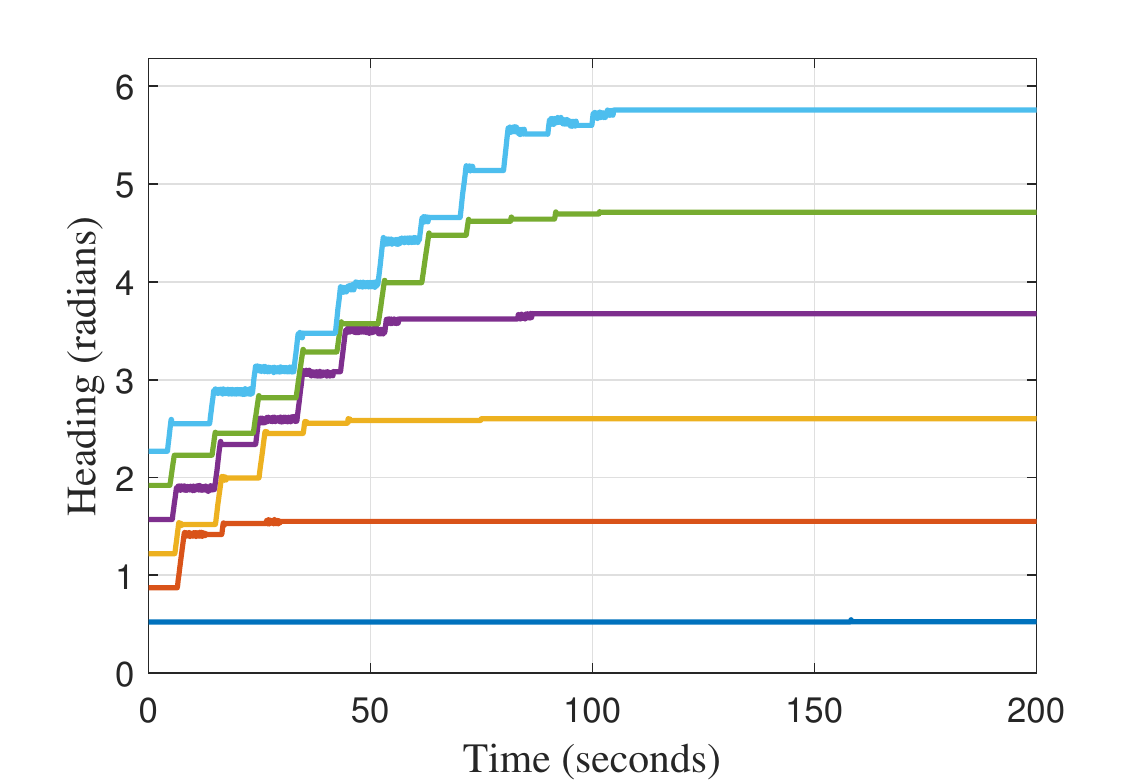}
	\centering
	\caption{Heading evolution of the robots under the heading desynchronization algorithm for \(N=6\) robots in an all-to-all topology, with \(l_{1} = 0.8\), \(l_{2} = 0\), \(\omega_0 = \frac{\pi}{5}\), and \(\omega_{\max} = 0.3 \omega_0\).}
	\label{fig:DesyncCFM_Heading_L2=0}
\end{figure}

\begin{figure}[t] 
	\includegraphics[width=\columnwidth]{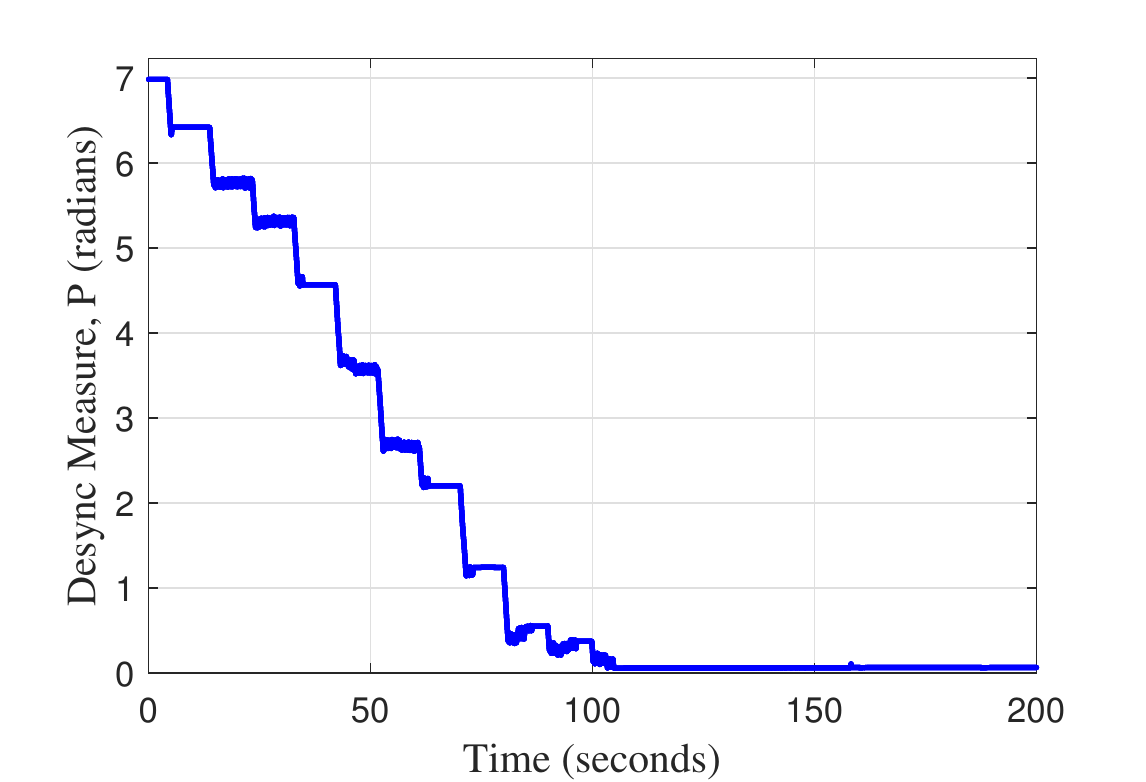}
	\centering
	\caption{Desynchronization measure, \(P\), as a function of time for the robot headings in Fig. \ref{fig:DesyncCFM_Heading_L2=0}.}
	\label{fig:DesyncCFM_ContainingArc_L2=0}
\end{figure}

\section{Conclusions} \label{sec:conclusion} 

In this paper, we consider the problem of decentralized heading control in mobile robots while considering the constraint of a finite rate of robot heading adjustment. We propose a decentralized hybrid-dynamical framework inspired by the study of pulse-coupled oscillators. Our approach avoids the discretization of a continuous control strategy while accounting for rate constraints by integrating the discrete-time communication and continuous-time control into the hybrid framework.

We rigorously analyze the hybrid dynamics of this pulse-coupled oscillator inspired heading control algorithm. We prove that this control strategy can achieve heading synchronization or desynchronization in a decentralized robotic network while conforming to the heading rate constraint. We demonstrate the control approach with physical experiments on a robotic platform.

We analyze the behavior of the robotic network using two pulse-coupled oscillator algorithms under the necessary heading rate constraint. Both of these algorithms have been used to study the behavior of oscillators with non-identical frequencies \cite{WangOptimalPRCSync2012, GaoPRCDesync2017}. Future research is required to study the behavior of oscillators having non-identical frequencies and non-identical robots with heterogeneous heading constraints. Additional synchronization and desynchronization algorithms for pulse-coupled oscillators have also been proposed \cite{Werner-Allen2005, ScalableSync2005, Nagpal2007, Scaglione2010, AngleaDesync2017}. Further research is needed to determine if these other algorithms can achieve the desired robotic heading under the heading rate constraint and proposed control strategy and if our results can be extended to achieve general motion coordination.

\section{Acknowledgment} \label{sec:acknowledge}
We would like to thank those who reviewed and provided feedback on the initial drafts of this paper.

Special thanks goes to Martyn Lemon, who contributed significantly to the experimental results in this paper.

\appendices

\section{Proof of Lemma \ref{lem:PhaseDesyncGeneral}} \label{ap:PhaseDesyncGeneralProof}
\begin{proof}
The proof of Lemma \ref{lem:PhaseDesyncGeneral} follows the approach given in \cite{GaoPRCDesync2017}, which gives a proof for a specific case, when the forward and backward coupling strengths are constant over their respective domains. Thus, we need to show that the desynchronization measure, \(P\), in \eqref{eq:DesyncIndex} decreases during each cycle until phase desynchronization is achieved.

To analyze the change of \(P\) at each firing instance, we will utilize the definitions of ``active pulse'' and ``silent pulse'' used in \cite{GaoPRCDesync2017}.

\begin{definition} \label{def:ActivePulse}
	A pulse is an ``active pulse'' when at least one robot has phase variable \(\phi_i \in (0,\frac{2\pi}{N}) \cup (2\pi\frac{N-1}{N},2\pi)\) when the pulse is emitted.
\end{definition}
\begin{definition} \label{def:SilentPulse}
	A pulse is a ``silent pulse'' when no robots have phase variable \(\phi_i \in (0,\frac{2\pi}{N}) \cup (2\pi\frac{N-1}{N},2\pi)\) when the pulse is emitted.
\end{definition}

According to Definitions \ref{def:ActivePulse} and \ref{def:SilentPulse}, a pulse is either an ``active pulse'' or a ``silent pulse''. For a ``silent pulse'', no robot phase variables are adjusted, according to \eqref{eq:PRFGeneral}. Thus, a ``silent pulse'' has no effect on phase differences in \eqref{eq:PhaseDifference} and the measure \(P\). Alternatively, an ``active pulse'' will cause robot phase variables to adjust, and thus may change the value of the measure \(P\). We can verify the existence of an ``active pulse'' in each cycle of pulse firings before phase desynchronization is achieved using the results found in \cite{GaoPRCDesync2017}. Thus, we need to show that the desynchronization measure \(P\) will decrease at each ``active pulse'' and converge to zero.

Without loss of generality, we assume that robot \(k\) emits an ``active pulse'' at time \(t = t_{k}\). From Definition \ref{def:ActivePulse}, there must be at least one robot with a phase within \((0,\frac{2\pi}{N}) \cup (2\pi\frac{N-1}{N},2\pi)\) at time \(t = t_{k}\). Without loss of generality, we assume that there are \(M\) robots with phase within \((0, \frac{2\pi}{N})\) and \(S\) robots with phase within \((2\pi\frac{N-1}{N},2\pi)\), where \(M\) and \(S\) are positive integers satisfying \(2 \leq M + S \leq N - 1\). The \(M\) and \(S\) phase variables are represented as \(\phi_{\widehat{k-1}}, \dots, \phi_{\widehat{k-M}}\) and \(\phi_{\widehat{k+1}}, \dots, \phi_{\widehat{k+S}}\) respectively, where the superscript ``\(\widehat{\ \ \ }\)'' represent modulo operation on \(N\), i.e., \(\widehat{\cdot} \triangleq (\cdot)\bmod N\). (Note: \(0\) maps to \(N\).) As a result of our assumptions, we have \(\phi_{\widehat{k-M}} < \frac{2\pi}{N} \leq \phi_{\widehat{k-M-1}}\) and \(\phi_{\widehat{k+S+1}} < 2\pi\frac{N-1}{N} \leq \phi_{\widehat{k+S}}\). 

Since \(\phi_{\widehat{k-1}}, \dots, \phi_{\widehat{k-M}}\) and \(\phi_{\widehat{k+1}}, \dots, \phi_{\widehat{k+S}}\) reside in \((0,\frac{2\pi}{N}) \cup (2\pi\frac{N-1}{N},2\pi)\), those robots will update their phase variables after receiving a pulse from robot \(k\) according to the phase update rule in \eqref{eq:PhaseUpdateGeneral}, such that
\begin{subequations} \label{eq:MSPhaseUpdate}
	\begin{equation} \label{eq:MSPhaseUpdate_1}
	\phi_{\widehat{k-i}}^{+} = (1 - L_{1}(\phi_{\widehat{k-i}}))\phi_{\widehat{k-i}} + L_{1}(\phi_{\widehat{k-i}})\frac{2\pi}{N}
	\end{equation}
	for \(i = 1, \dots, M\), and
	\begin{equation} \label{eq:MSPhaseUpdate_2}
	\phi_{\widehat{k+j}}^{+} = (1 - L_{2}(\phi_{\widehat{k+j}}))\phi_{\widehat{k+j}} + L_{2}(\phi_{\widehat{k+j}})2\pi\frac{N-1}{N}
	\end{equation}
	for \(j = 1, \dots, S\). Note that we also have \(\phi_{\widehat{k+q}}^{+} = \phi_{\widehat{k+q}}\) for \(q = S+1, \dots, N-M-1\) from \eqref{eq:PhaseUpdateGeneral}, and \(\phi_{k}^{+} = 0\).
\end{subequations}

According to \eqref{eq:PhaseDifference}	and \eqref{eq:MSPhaseUpdate}, the resulting phase variable differences after the phase update, \(\Delta_{i}^{+}\), 
caused by an ``active pulse'' from robot \(k\) can be given in seven parts by
\begin{subequations} \label{eq:PhaseDifferenceUpdate}
	\begin{multline} \label{eq:PhaseDifferenceUpdate_1}
	\Delta_{\widehat{k-M-1}}^{+} = \phi_{\widehat{k-M-1}}^{+} - \phi_{\widehat{k-M}}^{+} \\ = \phi_{\widehat{k-M-1}} - (1 - L_{1}(\phi_{\widehat{k-M}}))\phi_{\widehat{k-M}} -  L_{1}(\phi_{\widehat{k-M}})\frac{2\pi}{N}
	\end{multline}
	\begin{multline} \label{eq:PhaseDifferenceUpdate_2}
	\Delta_{\widehat{k-i}}^{+} = \phi_{\widehat{k-i}}^{+} - \phi_{\widehat{k-i+1}}^{+} \\
	= \big[(1 - L_{1}(\phi_{\widehat{k-i}}))\phi_{\widehat{k-i}} - (1 - L_{1}(\phi_{\widehat{k-i+1}}))\phi_{\widehat{k-i+1}}\big] \\
	+ \big(L_{1}(\phi_{\widehat{k-i}}) - L_{1}(\phi_{\widehat{k-i+1}})\big)\frac{2\pi}{N}
	\end{multline}
	for \(i = 2, \dots, M\),
	\begin{multline} \label{eq:PhaseDifferenceUpdate_3}
	\Delta_{\widehat{k-1}}^{+} = \phi_{\widehat{k-1}}^{+} - \phi_{\widehat{k}}^{+} \\
	= (1 - L_{1}(\phi_{\widehat{k-1}}))\phi_{\widehat{k-1}} +  L_{1}(\phi_{\widehat{k-1}})\frac{2\pi}{N}
	\end{multline}
	\begin{multline} \label{eq:PhaseDifferenceUpdate_4}
	\Delta_{\widehat{k}}^{+} = \phi_{\widehat{k}}^{+} - \phi_{\widehat{k+1}}^{+} + 2\pi \\
	= 2\pi - (1 - L_{2}(\phi_{\widehat{k+1}}))\phi_{\widehat{k+1}} -  L_{2}(\phi_{\widehat{k+1}})2\pi\frac{N-1}{N}
	\end{multline}
	\begin{multline} \label{eq:PhaseDifferenceUpdate_5}
	\Delta_{\widehat{k+j}}^{+} = \phi_{\widehat{k+j}}^{+} - \phi_{\widehat{k+j+1}}^{+}\\
	= \big[(1 - L_{2}(\phi_{\widehat{k+j}}))\phi_{\widehat{k+j}} - (1 - L_{2}(\phi_{\widehat{k+j+1}}))\phi_{\widehat{k+j+1}}\big] \\
	+ \big(L_{2}(\phi_{\widehat{k+j}}) - L_{2}(\phi_{\widehat{k+j+1}})\big)2\pi\frac{N-1}{N}
	\end{multline}
	for \(j = 1, \dots, S-1\),
	\begin{multline} \label{eq:PhaseDifferenceUpdate_6}
	\Delta_{\widehat{k+S}}^{+} = \phi_{\widehat{k+S}}^{+} - \phi_{\widehat{k+S+1}}^{+} \\
	= (1 - L_{2}(\phi_{\widehat{k+S}}))\phi_{\widehat{k+S}} + L_{2}(\phi_{\widehat{k+S}})2\pi\frac{N-1}{N} - \phi_{\widehat{k+S+1}}
	\end{multline}
	and
	\begin{equation} \label{eq:PhaseDifferenceUpdate_7}
	\Delta_{\widehat{k+q}}^{+} = \phi_{\widehat{k+q}}^{+} - \phi_{\widehat{k+q+1}}^{+} = \phi_{\widehat{k+q}} - \phi_{\widehat{k+q+1}} = \Delta_{\widehat{k+q}}
	\end{equation}
	for \(q = S+1, \dots, N-M-2\).
\end{subequations}
Note that \eqref{eq:PhaseDifferenceUpdate} can be simplified to that found in \cite{GaoPRCDesync2017} if the forward and backward coupling functions are constant over their respective domains.

The new value for \(P\) (denoted as \(P^{+}\)) after the update is given by:
\begin{equation} \label{eq:DesyncIndexUpdate}
P^{+} = \sum_{k=1}^{N} |\Delta_{k}^{+} - \frac{2\pi}{N}|
\end{equation}

To determine the change in the value of measure \(P\) caused by the ``active pulse'' from robot \(k\), we can calculate the difference between the value of \(P\) before and after the phase update:
\begin{equation} \label{eq:DesyncIndexDifference}
P^{+} - P = \sum_{k=1}^{N} \Big[|\Delta_{k}^{+} - \frac{2\pi}{N}| - |\Delta_{k} - \frac{2\pi}{N}|\Big]
\end{equation}

We can divide \eqref{eq:DesyncIndexDifference} into seven parts, using the expressions found in \eqref{eq:PhaseDifference} and \eqref{eq:PhaseDifferenceUpdate}:
\begin{subequations} \label{eq:DesyncIndexDifferenceExpand}
	\begin{equation*}
	P^{+} - P = \sum_{k=1}^{N} \Big[|\Delta_{k}^{+} - \frac{2\pi}{N}| - |\Delta_{k} - \frac{2\pi}{N}|\Big]
	\end{equation*}
	\begin{equation} \label{eq:DesyncIndexDifferenceExpand_p1}
	= |\Delta_{\widehat{k-M-1}}^{+} - \frac{2\pi}{N}| - |\Delta_{\widehat{k-M-1}} - \frac{2\pi}{N}|
	\end{equation}
	\begin{equation} \label{eq:DesyncIndexDifferenceExpand_p2}
	+ \sum_{i=2}^{M} \Big[|\Delta_{\widehat{k-i}}^{+} - \frac{2\pi}{N}| - |\Delta_{\widehat{k-i}} - \frac{2\pi}{N}|\Big]
	\end{equation}
	\begin{equation} \label{eq:DesyncIndexDifferenceExpand_p3}
	+ |\Delta_{\widehat{k-1}}^{+} - \frac{2\pi}{N}| - |\Delta_{\widehat{k-1}} - \frac{2\pi}{N}|
	\end{equation}
	\begin{equation} \label{eq:DesyncIndexDifferenceExpand_p4}
	+ |\Delta_{k}^{+} - \frac{2\pi}{N}| - |\Delta_{k} - \frac{2\pi}{N}|
	\end{equation}
	\begin{equation} \label{eq:DesyncIndexDifferenceExpand_p5}
	+ |\sum_{j=1}^{S-1} \Big[|\Delta_{\widehat{k+j}}^{+} - \frac{2\pi}{N}| - |\Delta_{\widehat{k+j}} - \frac{2\pi}{N}|\Big]
	\end{equation}
	\begin{equation} \label{eq:DesyncIndexDifferenceExpand_p6}
	+ |\Delta_{\widehat{k+S}}^{+} - \frac{2\pi}{N}| - |\Delta_{\widehat{k+S}} - \frac{2\pi}{N}|
	\end{equation}
	\begin{equation} \label{eq:DesyncIndexDifferenceExpand_p7}
	+ \sum_{q=S+1}^{N-M-2} \Big[|\Delta_{\widehat{k+q}}^{+} - \frac{2\pi}{N}| - |\Delta_{\widehat{k+q}} - \frac{2\pi}{N}|\Big]
	\end{equation}
\end{subequations}

Eq. \eqref{eq:DesyncIndexDifferenceExpand_p2} can be simplified as follows:
\begin{multline} \label{eq:DesyncIndexDifferenceExpand_p2simple}
\sum_{i=2}^{M} \Big[|\big[(1 - L_{1}(\phi_{\widehat{k-i}}))\phi_{\widehat{k-i}} - (1 - L_{1}(\phi_{\widehat{k-i+1}}))\phi_{\widehat{k-i+1}}\big] \\
+ \big(L_{1}(\phi_{\widehat{k-i}}) - L_{1}(\phi_{\widehat{k-i+1}})\big)\frac{2\pi}{N} - \frac{2\pi}{N}| \\
- |(\phi_{\widehat{k-i}} - \phi_{\widehat{k-i+1}}) - \frac{2\pi}{N}|\Big] \\
= \sum_{i=2}^{M} \Big[\frac{2\pi}{N} \big(L_{1}(\phi_{\widehat{k-i+1}}) - L_{1}(\phi_{\widehat{k-i}})\big) \\
+ L_{1}(\phi_{\widehat{k-i}})\phi_{\widehat{k-i}} - L_{1}(\phi_{\widehat{k-i+1}})\phi_{\widehat{k-i+1}} \Big] \\
= \frac{2\pi}{N} \big(L_{1}(\phi_{\widehat{k-1}}) - L_{1}(\phi_{\widehat{k-M}})\big) \\
+ L_{1}(\phi_{\widehat{k-M}})\phi_{\widehat{k-M}} - L_{1}(\phi_{\widehat{k-1}})\phi_{\widehat{k-1}}
\end{multline}
where we used the relationships \(\phi_{\widehat{k-i}}-\phi_{\widehat{k-i+1}} < \frac{2\pi}{N}\) and \(\big[(1 - L_{1}(\phi_{\widehat{k-i}}))\phi_{\widehat{k-i}} - (1 - L_{1}(\phi_{\widehat{k-i+1}}))\phi_{\widehat{k-i+1}}\big] +  \frac{2\pi}{N}(L_{1}(\phi_{\widehat{k-i}}) - L_{1}(\phi_{\widehat{k-i+1}}))< \frac{2\pi}{N}\) for \(i = 2, \dots, M\).

Eq. \eqref{eq:DesyncIndexDifferenceExpand_p3} can be simplified as follows:
\begin{multline} \label{eq:DesyncIndexDifferenceExpand_p3simple}
|(1 - L_{1}(\phi_{\widehat{k-1}}))\phi_{\widehat{k-1}} + L_{1}(\phi_{\widehat{k-1}})\frac{2\pi}{N} - \frac{2\pi}{N}| - |\phi_{\widehat{k-1}} -  \frac{2\pi}{N}| \\
= L_{1}(\phi_{\widehat{k-1}})(\phi_{\widehat{k-1}}-\frac{2\pi}{N})
\end{multline}
where we used the fact that \(\phi_{\widehat{k-1}} < \frac{2\pi}{N}\).

Eq. \eqref{eq:DesyncIndexDifferenceExpand_p4} can be simplified as follows:
\begin{multline} \label{eq:DesyncIndexDifferenceExpand_p4simple}
|2\pi - (1 - L_{2}(\phi_{\widehat{k+1}}))\phi_{\widehat{k+1}} - L_{2}(\phi_{\widehat{k+1}})2\pi\frac{N-1}{N} - \frac{2\pi}{N}| \\
- |(2\pi - \phi_{\widehat{k+1}}) -  \frac{2\pi}{N}| \\
= L_{2}(\phi_{\widehat{k+1}})(2\pi\frac{N-1}{N} - \phi_{\widehat{k+1}})
\end{multline}
where we used the fact that \(2\pi - \phi_{\widehat{k+1}} < \frac{2\pi}{N}\).

Eq. \eqref{eq:DesyncIndexDifferenceExpand_p5} can be simplified as follows:
\begin{multline} \label{eq:DesyncIndexDifferenceExpand_p5simple}
\sum_{j=1}^{S-1} \Big[|\big[(1 - L_{2}(\phi_{\widehat{k+j}}))\phi_{\widehat{k+j}} - (1 - L_{2}(\phi_{\widehat{k+j+1}}))\phi_{\widehat{k+j+1}}\big] \\
+ \big(L_{2}(\phi_{\widehat{k+j}}) - L_{2}(\phi_{\widehat{k+j+1}})\big)2\pi\frac{N-1}{N} - \frac{2\pi}{N}| \\
- |(\phi_{\widehat{k+j}} - \phi_{\widehat{k+j+1}}) - \frac{2\pi}{N}|\Big] \\
= \sum_{j=1}^{S-1} \Big[2\pi\frac{N-1}{N} \big(L_{2}(\phi_{\widehat{k+j+1}}) - L_{2}(\phi_{\widehat{k+j}})\big) \\
+ L_{2}(\phi_{\widehat{k+j}})\phi_{\widehat{k+j}} - L_{2}(\phi_{\widehat{k+j+1}})\phi_{\widehat{k+j+1}} \Big] \\
= 2\pi\frac{N-1}{N} \big(L_{2}(\phi_{\widehat{k+S}}) - L_{2}(\phi_{\widehat{k+1}})\big) \\
+ L_{2}(\phi_{\widehat{k+1}})\phi_{\widehat{k+1}} - L_{2}(\phi_{\widehat{k+S}})\phi_{\widehat{k+S}}
\end{multline}
where we used the relationships \(\phi_{\widehat{k+j}}-\phi_{\widehat{k+j+1}} < \frac{2\pi}{N}\) and \(\big[(1 - L_{2}(\phi_{\widehat{k+j}}))\phi_{\widehat{k+j}} - (1 - L_{2}(\phi_{\widehat{k+j+1}}))\phi_{\widehat{k+j+1}}\big] + \big(L_{2}(\phi_{\widehat{k+j}}) - L_{2}(\phi_{\widehat{k+j+1}})\big)2\pi\frac{N-1}{N} < \frac{2\pi}{N}\) for \(j = 1, \dots, S-1\).

Eq. \eqref{eq:DesyncIndexDifferenceExpand_p7} can be simplified as follows:
\begin{multline} \label{eq:DesyncIndexDifferenceExpand_p7simple}
\sum_{q=S+1}^{N-M-2} \Big[|\Delta_{\widehat{k+q}}^{+} - \frac{2\pi}{N}| - |\Delta_{\widehat{k+q}} - \frac{2\pi}{N}|\Big] \\
= \sum_{q=S+1}^{N-M-2} \Big[|\Delta_{\widehat{k+q}} - \frac{2\pi}{N}| - |\Delta_{\widehat{k+q}} - \frac{2\pi}{N}|\Big] = 0
\end{multline}
Thus, combining \eqref{eq:DesyncIndexDifferenceExpand_p1}, \eqref{eq:DesyncIndexDifferenceExpand_p6}, and \eqref{eq:DesyncIndexDifferenceExpand_p2simple}--\eqref{eq:DesyncIndexDifferenceExpand_p7simple}, we can write \eqref{eq:DesyncIndexDifference} in two main parts:
\begin{subequations} \label{eq:DesyncIndexDifferenceFinal}
	\begin{equation*}
	P^{+} - P = \sum_{k=1}^{N} \Big[|\Delta_{k}^{+} - \frac{2\pi}{N}| - |\Delta_{k} - \frac{2\pi}{N}|\Big]
	\end{equation*}
	\begin{multline} \label{eq:DesyncIndexDifferenceFinal_pA}
	= |\Delta_{\widehat{k-M-1}}^{+} - \frac{2\pi}{N}| - |\Delta_{\widehat{k-M-1}} - \frac{2\pi}{N}| \\
	+ L_{1}(\phi_{\widehat{k-M}}) \big(\phi_{\widehat{k-M}} - \frac{2\pi}{N}\big)
	\end{multline}
	\begin{multline} \label{eq:DesyncIndexDifferenceFinal_pB}
	+ |\Delta_{\widehat{k+S}}^{+} - \frac{2\pi}{N}| - |\Delta_{\widehat{k+S}} - \frac{2\pi}{N}| \\
	+ L_{2}(\phi_{\widehat{k+S}}) \big(2\pi\frac{N-1}{N} - \phi_{\widehat{k+S}}\big)
	\end{multline}
\end{subequations}

This result in \eqref{eq:DesyncIndexDifferenceFinal} is similar to that found in \cite{GaoPRCDesync2017}. Since \(L_{1}(\phi_{\widehat{k-M}})\) and \(L_{2}(\phi_{\widehat{k+S}})\) are constants in the range \([0,1)\), then a direct comparison can be made with \textbf{Part A} and \textbf{Part B} in \cite{GaoPRCDesync2017} to \eqref{eq:DesyncIndexDifferenceFinal_pA} and \eqref{eq:DesyncIndexDifferenceFinal_pB}, respectively. Borrowing the conclusion made in \cite{GaoPRCDesync2017}, we conclude that \(P^{+} - P \leq 0\) holds and cannot always be zero, and thus the measure \(P\) in \eqref{eq:DesyncIndex} will decrease to zero. Therefore, the network will achieve phase desynchronization using a phase response function given by \eqref{eq:PRFGeneral} 
such that the forward coupling function, \(L_{1}(\phi_i) \in [0,1)\), over the domain \((0,\frac{2\pi}{N})\) and the backward coupling function, \(L_{2}(\phi_i) \in [0,1)\), over the domain \((2\pi\frac{N-1}{N},2\pi)\) cause the phase update rule in \eqref{eq:PhaseUpdateGeneral} to be strictly increasing.
\end{proof}

\bibliographystyle{unsrt}
\bibliography{PRCHeading_Rate_Constraint_double_column.bib}

\end{document}